\documentclass[12pt]{amsart}

\usepackage[margin=1in]{geometry}
\usepackage[utf8]{inputenc}
\usepackage{amsmath,amsthm,amssymb,textcomp,gensymb,graphicx,float,geometry,mathrsfs,tikz,tikz-cd,tikz-3dplot,upgreek,enumitem,moreenum,subfig,multirow}
\usepackage[title]{appendix}
\usepackage[all]{xy}

\usepackage[
backend=biber,
style=alphabetic,
sorting=nyt
]{biblatex}
\usepackage{hyperref}
\hypersetup{
colorlinks=true,
allcolors=blue
}

\usepackage{etoolbox}
\cslet{blx@noerroretextools}\empty
\usepackage{autonum}

\makeatletter 
\renewcommand*\env@matrix[1][\arraystretch]{%
  \edef\arraystretch{#1}%
  \hskip -\arraycolsep
  \let\@ifnextchar\new@ifnextchar
  \array{*\c@MaxMatrixCols c}}
\makeatother

\makeatletter 
\newcommand{\setword}[2]{%
  \phantomsection
  #1\def\@currentlabel{\unexpanded{#1}}\label{#2}%
}
\makeatother

\makeatletter 
\newtheorem*{rep@theorem}{\rep@title}
\newcommand{\newreptheorem}[2]{%
\newenvironment{rep#1}[1]{%
 \def\rep@title{#2 \ref{##1}}%
 \begin{rep@theorem}}%
 {\end{rep@theorem}}}
\makeatother

\def\Z{{\mathbb Z}}

\def\B{{\mathbb B}}
\def\CC{{\mathbb C}}

\def\JJ{{\mathbb J}}

\def\NN{{\mathbb N}}

\def\PP{{\mathbb P}}

\def\UU{{\mathbb U}}

\def\cal{\mathcal}

\def\cB{{\cal B}}
\def\cC{{\cal C}}
\def\cD{{\cal D}}
\def\cE{{\cal E}}

\def\cH{{\cal H}}

\def\cL{{\cal L}}

\def\cO{{\cal O}}

\def\spa{\operatorname{span}}

\def\0ol{{\bar 0}}
\def\1ol{{\bar 1}}
\def\2ol{{\bar 2}}
\def\ol2{{\bar 2}}
\def\3ol{{\bar 3}}
\def\4ol{{\bar 4}}
\def\5ol{{\bar 5}}
\def\6ol{{\bar 6}}
\def\7ol{{\bar 7}}
\def\8ol{{\bar 8}}
\def\9ol{{\bar 9}}

\def\P2Skly{\PP^2_{Skly}}

\def\mult{\operatorname{mult}} 
\def\dom{\operatorname{dom}} 

\def\ker{\operatorname {ker}}

\def\pd{{\operatorname {\partial}}}

\def\Tr{\operatorname {Tr}}

\def\det{\operatorname{det}}

\def\dim{\operatorname{dim}}

\def\Im{\operatorname{Im}}

\def\min{\operatorname{min}}

\def\pd{{\partial}}

\def\rank{\operatorname{rank}}

\def\sup{\operatorname{sup}}

\def\ul1{\operatorname{\underline{1}}}

\def\dirlim{\mathop{\vtop{\baselineskip -100pt\lineskip -1pt\lineskiplimit 0pt
\setbox0\hbox{lim}\copy0\hbox to \wd0{\rightarrowfill}}}\limits}
\def\invlim{\mathop{\vtop{\baselineskip -100pt\lineskip -1pt\lineskiplimit 0pt
\setbox0\hbox{lim}\copy0\hbox to \wd0{\leftarrowfill}}}\limits}

\def\I11{{1 \kern -0.8pt \! \mbox{l}}}
\def\mumu{{\mu\kern-4.2pt\mu}}
\def\bfmu{{\mu\kern-4.2pt\mu}}
\def\2slash{\backslash \! \backslash}

\def\boxtimes{\setbox0\hbox{$\Box$}\copy0\kern-\wd0\hbox{$\times$}}

\theoremstyle{plain}
\newtheorem{theorem}{Theorem}
\newreptheorem{theorem}{Theorem} 
\newtheorem{theorem*}{Theorem}
\newtheorem{definition}{Definition}
\newtheorem{corollary}{Corollary}
\newtheorem{proposition}{Proposition}
\newtheorem{lemma}{Lemma}

\theoremstyle{remark}

\newcommand{\C}{\mathbb{C}}
\newcommand{\N}{\mathbb{N}}

\numberwithin{equation}{section}
\numberwithin{lemma}{section}
\numberwithin{proposition}{section}
\numberwithin{corollary}{section}

\begin{document}

\title{Band Spectrum Singularities for Schr\"odinger Operators}
\author{Alexis Drouot, Curtiss Lyman}
\date{\today}

\begin{abstract}
      In this paper, we develop a systematic framework to study the dispersion surfaces of Schr{\"o}dinger operators $ -\Delta + V$, where the potential $V \in C^\infty(\mathbb{R}^n,\mathbb{R})$ is periodic with respect to a lattice $\Lambda \subset \mathbb{R}^n$ and respects the symmetries of $\Lambda$. Our analysis combines the theory of holomorphic families of operators of type (A) with the seminal work of Fefferman--Weinstein \cite{feffer12}. It allows us to extend results on the existence of spectral degeneracies past a perturbative regime. As an application, we describe the generic structure of some singularities in the band spectrum of Schr\"odinger operators invariant under the three-dimensional simple, body-centered and face-centered cubic lattices. 
 \end{abstract}

 \maketitle

\section{Introduction}\label{chp:1}

Analyzing the behavior of waves in periodic structures is a central theme in condensed matter physics, electromagnetism and photonics. This includes for instance  electronic conduction: the flow of electrons through a crystal. In the framework of quantum mechanics, these waves solve the time-dependent Schr{\" o}dinger equation
\begin{equation}\label{eq:1}
    i\pd_t \psi = (-\Delta + V)\psi, \qquad \text{where:}
\end{equation}
\begin{itemize}
    \item the potential $V$ is periodic with respect to a lattice $\Lambda$;
    \item the function $\psi$ is the wavefunction of the electron, i.e. $|\psi(t,x)|^2$ is the density of probability of finding the electron at position $x$, at time $t$.
\end{itemize}
Solutions of \eqref{eq:1} can be written as superpositions of time-harmonic waves: functions of the form $e^{-i\mu t}\phi(x)$, where $\phi$ and $\mu$ solve the \textit{eigenvalue problem}
\begin{equation}\label{eq:9a}
    \mu\phi = (-\Delta + V)\phi.
\end{equation}
The equation \eqref{eq:9a} will be the main focus of this work.

Because the potential $V$ is periodic, the operator $-\Delta+V$ has absolutely continuous spectrum on $L^2$, see \cite[Theorem XIII.100]{reed}. The corresponding \textit{generalized eigenstates} are superpositions over $k \in \mathbb{R}^n$ of Floquet--Bloch modes: solutions $\phi(x;k)$ to
\begin{equation}\label{eq:9b}
  \begin{gathered}
  (-\Delta +V)\phi(x;k)  = \mu(k)\phi(x;k), \quad x \in \mathbb{R}^n,\\
    \phi(x+v;k) = e^{ik\cdot v}\phi(x;k), \quad v \in \Lambda.
    \end{gathered}
\end{equation}
For each $k \in \mathbb{R}^n$, the problem \eqref{eq:9b} has a discrete set of solutions $\mu(k)$, which corresponds to the spectrum of $-\Delta+V$ on the space of quasiperiodic functions
\begin{equation}\label{eq:l2k}
    L^2_k = \big\{f \in L^2_\mathrm{loc}(\mathbb{R}^n,\C):\;f(x+v) = e^{ik\cdot v}f(x), \ v \in \Lambda\big\}.
    \end{equation}
The maps $k \mapsto \mu(k)$ are called \textit{dispersion surfaces;} brought together they form the \textit{band spectrum} of $-\Delta+V$. The local properties of these maps control the effective dynamics of wavepackets \cite{allaire}, and singularities in the band spectrum trigger unusual behavior of waves. For instance, Dirac cones -- conical intersection of dispersion surfaces in 2D -- give rise to Dirac-like propagation of wavepackets: this explains the relativistic behavior of electrons observed in graphene  \cite{feffer12,feffer14}.

The mathematical analysis of band spectrum singularities started with the seminal work of Fefferman--Weinstein \cite{feffer12}, who proved genericity of Dirac points in honeycomb lattices. This has sparked various mathematical investigations of spectral degeneracies in other two-dimensional lattices \cite{finite, lieb, photonic, cassier, chaban}; the only three-dimensional work so far is \cite{guo}. These works share a common strategy, split in two parts: 
\begin{itemize}
    \item[(a)] proving results for small potentials via perturbation theory and symmetry arguments; 
    \item[(b)] extending them to generic potentials using an analyticity argument due to \cite{feffer12}.
\end{itemize}

To prove (b), the above works have referred to  \cite{feffer12} for details. This motivates the development of a general formalism under which  one can apply the Fefferman--Weinstein theory. In particular, \cite{guo} proved that the band spectrum of Schr\"odinger operator with \textit{small} potentials, periodic with respect to the body-centered cubic lattice and invariant under the octahedral group, presents a three-fold Weyl point -- see Definition \ref{def:1} and Figure \ref{fig:1}. In \cite[\S 5.2]{guo}, they \textit{conjectured} that this extends to \textit{large} potentials; we prove this statement here using the theory of holomorphic families of operators of type (A) \cite{kato, rellich}. In addition, we provide applications of our approach to the generic band spectrum singularities of 3D Schr\"odinger operators invariant under the simple, face-centered and body-centered cubic lattices.

\subsection{Main results} We formulate here our two main results. The first one, together with the results from Section \ref{sec:3.4}, gives us a systematic framework for the generic analysis of dispersion surfaces of Schr\"odinger operators of the form 
\begin{equation}\label{eq:Hz}
    H_z = -\Delta + zV,
\end{equation}
where the potential $V$ is assumed to be periodic with respect to a lattice $\Lambda$. Our second main result then applies this framework to study the band spectrum singularities of Schr\"odinger operators invariant under cubic lattices.

\begin{theorem}\label{thm:1new}
Let $z \in \mathbb{R}$, $\mu(z)$ an $L^2_K$-eigenvalue of $H_z$, $\pi(z): L^2_K \rightarrow L^2_K$ be the corresponding eigenprojector, and $\cE(z)$ be its range. There exist $\varepsilon,\delta,C > 0$ such that for $\|\kappa\|< \varepsilon$, the $L^2_{K+\kappa}$ eigenvalues of $H_z$ in $\B_\delta(\mu(z))$ satisfy
\begin{equation}\label{eq:det}
    \det\big((u(z)-\mu) + M(z,\kappa) + R(\mu,\kappa)\big)|_{\cE(z)}=0, \qquad \qquad \text{where:}
\end{equation}
\begin{itemize}
    \item $M(z,\kappa)$ is the operator $=-\pi(z) (2i\kappa\cdot\nabla)\pi(z)$ on $\cE(z)$; and
    \item $R(\mu,\kappa)$ is an operator on $\cE(z)$ that satisfies $\|R(\mu,\kappa)\| \leq C\|\kappa\|^2$
\end{itemize}

Furthermore, if $\mu(z)$ depends analytically on $z$, then the characteristic polynomial of $M(z,\kappa)$ 
also depends analytically on $z$.
\end{theorem}

In particular, Theorem \ref{thm:1new} reduces the Floquet-Bloch problem \eqref{eq:floquet} to the finite-dimensional characteristic value problem \eqref{eq:det}. In addition, the theory of holomorphic families of operators of type (A) (which we review in \S\ref{chp:2}) ensures that we can represent all the eigenvalues of $H_z$ by analytic functions on $\mathbb{R}$. This means that if $\mu_0$ is an eigenvalue of $H_{z_0}$ for some $z_0 \in \mathbb{R}$, then there exists a function $\mu(z)$, analytic on $\mathbb{R}$, such that $\mu(z)$ is an eigenvalue of $H_z$ for all $z \in \mathbb{R}$ and $\mu(z_0) = \mu_0$. 
Later, we will classify the types of band-spectrum degeneracies depending on which coefficients of the characteristic polynomial of $M(z,\kappa)$ vanish. Because analytic functions either vanish identically or only on a discrete set, the second part of Theorem \ref{thm:1new} will guarantee that the spectral degeneracies identified for small values of $z$ must persist for generic values of $z$. 

We apply Theorem \ref{thm:1new} to Schr{\" o}dinger operators $H_z$ with a potential $V$ invariant under the symmetries of the lattice: for instance, $2\pi/3$-invariant and even in the case of the honeycomb lattice; or invariant under the octahedral group for cubic lattices. This is because additional symmetries come with higher multiplicities of eigenvalues, which in turn translate to singularities in the band spectrum. At the dynamical levels, waves initially localized in frequency near these high-degeneracy momenta exhibit anomalous propagation. For instance, for Schr\"odinger operators invariant under the honeycomb lattice, they are two-scale functions whose envelop effectively solve a Dirac equations \cite{feffer14}. As an explicit example, we study the band spectrum singularities of Schr\"odinger operators invariant under the three 3D cubic lattices: the simple, body-centered and face-centered cubic lattices -- see \S\ref{chp:3}-\ref{chp:4} for precise definitions and Figure \ref{fig:2} for visual representations of their Brillouin zone and of the results of Theorem \ref{thm:2}.

\begin{figure}[!b]
    \centering
    \begin{tabular}{|c|c|c|c|c|}
        \hline
        &  & \multicolumn{2}{|c|}{} & \\[-1em]
        Lattice & Simple & \multicolumn{2}{|c|}{Body-centered} & Face-centered \\
        \hline
        & & \multicolumn{2}{|c|}{} & \\
        Brillouin zone & \includegraphics[width=0.17\textwidth]{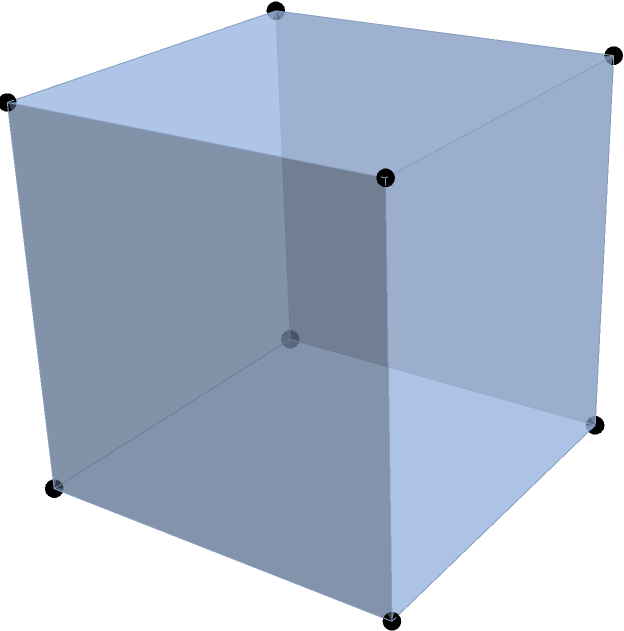} &\multicolumn{2}{|c|}{\includegraphics[width=0.17\textwidth]{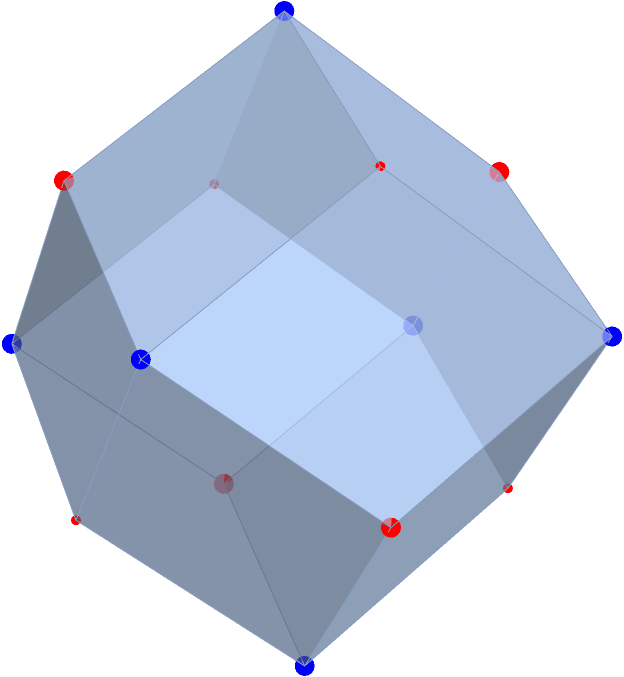}} & \includegraphics[width=0.17\textwidth]{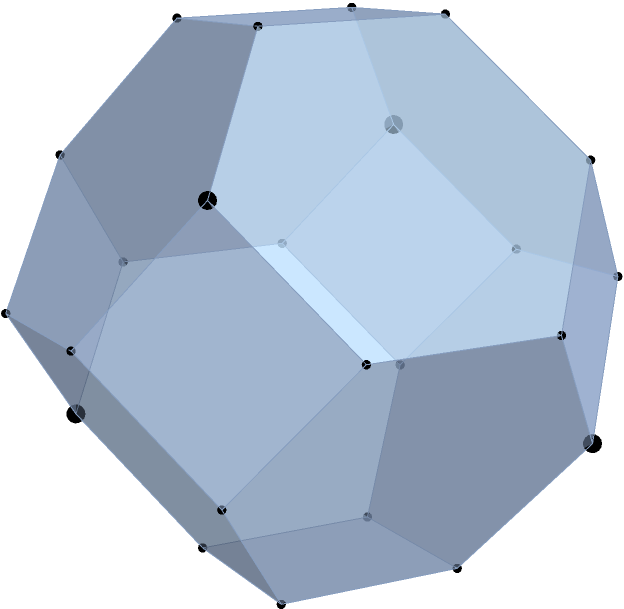} \\
        & & \multicolumn{2}{|c|}{} & \\
        \hline
        Degeneracy type &  Quadratic & \multicolumn{2}{|c|}{3-fold Weyl, quadratic}  & Basin   \\
        \hline
    \end{tabular}
    \caption{Spectral degeneracies in cubic lattices.}\label{fig:2}
\end{figure}

\begin{definition}\label{def:1} Let $E \in \mathbb{R}, K \in \mathbb{R}^3$. We say that a Schr\"odinger operator $H$ has:
\begin{itemize}
    \item An $m$-fold quadratic point at $(K,E)$ if $E$ is a $L^2_K$-eigenvalue of $H$ of multiplicity $m > 1$ and the Floquet-Bloch problem \eqref{eq:floquet} has $m$ solutions $\mu_1(k), \dots,  \mu_m(k)$:
    \begin{equation}
        \mu_j(K+\kappa) = E + \cO(\|\kappa\|^2), \quad j =1, \dots, m, \quad \kappa \rightarrow 0.
    \end{equation}
    
    \item A (two-fold) basin point at $(K,E)$ if $E$ is a double $L^2_K$-eigenvalue of $H$, and there exists some $v \neq 0 \in \mathbb{R}^3$ such that for $\kappa$ satisfying $v\cdot \kappa \not = 0$, the Floquet-Bloch problem \eqref{eq:floquet} has 2 solutions $\mu_+(k), \mu_-(k)$:
    \begin{equation}
        \mu_\pm(K+\kappa) = E \pm |v \cdot \kappa| + \cO(\|\kappa\|^2), \quad \kappa \rightarrow 0.
    \end{equation}

    \item A Weyl point at $(K,E)$ if $E$ is a double $L^2_K$-eigenvalue of $H$, and there exists some $\alpha \neq 0 \in \mathbb{R}$ such that the Floquet-Bloch problem \eqref{eq:floquet} has 2 solutions $\mu_+(k), \mu_-(k)$:
    \begin{equation}\label{eq:9c}
        \mu_\pm(K+\kappa) = E \pm \alpha \|\kappa\| + \cO(\|\kappa\|^2), \quad \kappa \rightarrow 0.
    \end{equation}
    
    \item A three-fold Weyl point at $(K,E)$ if $E$ is a triple $L^2_K$-eigenvalue of $H$, and there exists some $\alpha \neq 0 \in \CC$ such that the Floquet-Bloch problem \eqref{eq:floquet} has 3 solutions $\mu_1(k), \mu_2(k), \mu_3(k)$:
    \begin{equation}\label{eq:9d}
        \mu_j(K+\kappa) = E + \lambda_{\alpha,j}(\kappa) + \cO(\|\kappa\|^2), \quad j=1,2,3, \quad \kappa \rightarrow 0,
    \end{equation}
    where $\lambda_{\alpha,j}(\kappa)$ are the three roots of the polynomial $\lambda^3-4 |\alpha|^2 \|\kappa\|^2 \lambda + 16 \Im(\alpha^3)\kappa_1\kappa_2\kappa_3$.     
\end{itemize}
\end{definition}

\begin{figure}[b]
    \centering
    \includegraphics[width=0.7\textwidth]{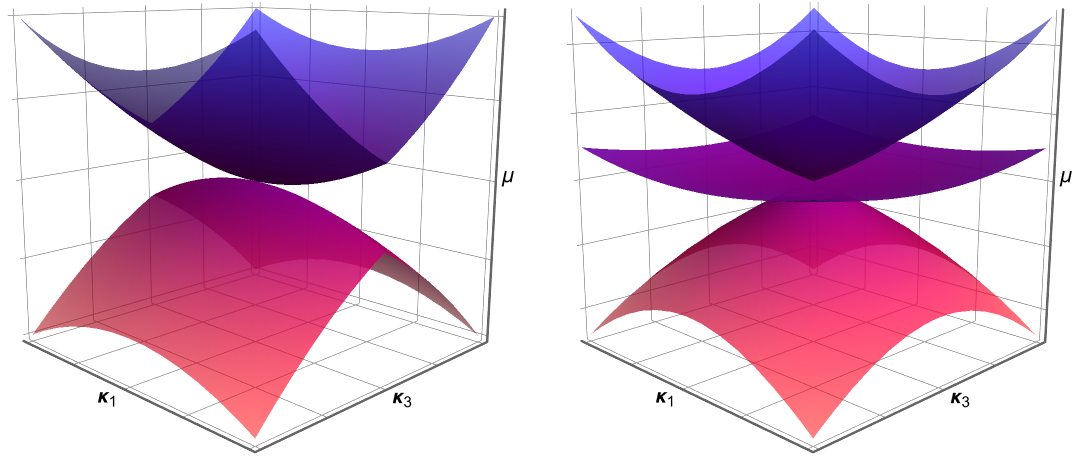}
    \caption{\label{fig:1} Possible cross sections of a basin point (left) and a three-fold Weyl point (right).}
    \label{fig:CubicDispersionSurfaces}
\end{figure}

Our second main result shows that Schr\"odinger operators periodic with respect to the cubic lattices and invariant under the octahedral group admit such degenerate points in their band spectrum.

\begin{theorem}\label{thm:2} For generic potentials $V \in C^\infty(\mathbb{R}^3,\mathbb{R})$ invariant under the octahedral group and periodic with respect to a cubic lattice $\Lambda \subset \mathbb{R}^3$, and generic values of $z \in \mathbb{R}$, the band spectrum of $H_z = -\Delta+z V$ has at least:
    \begin{enumerate}[label=(\arabic*)]
        \item[(i)] Two three-fold quadratic points if $\Lambda$ is the simple cubic lattice;
        \item[(ii)] One three-fold Weyl point as well as one two-fold and one three-fold quadratic point if $\Lambda$ is the body-centered cubic lattice;
        \item[(iii)] One basin point if $\Lambda$ is the face-centered cubic lattice.
    \end{enumerate}
\end{theorem}

In Theorem \ref{thm:2}, ``generic potentials $V \in C^\infty(\mathbb{R}^3,\mathbb{R})$" means all of $C^\infty(\mathbb{R}^3,\mathbb{R})$ but a finite union of hyperplanes, and ``generic values of $z \in \mathbb{R}$" means all of $\mathbb{R}$ away from a discrete set. In particular, Theorem \ref{thm:2} proves the conjecture formulated in \cite[\S5.2]{guo}.

\subsection{Proofs}\label{sec:1.2} 

The proof of Theorem \ref{thm:1new}, which we cover in \S \ref{chp:3}, is directly inspired by \cite{feffer12}, but the technical core of its proof relies on
different arguments. To study the dispersion surfaces near some quasi-momentum $K \in \mathbb{R}^n$, i.e. the $L^2_{K+\kappa}$-eigenvalues of $H_z$ for $\kappa$ small, both approaches instead analyze the $L^2_K$-eigenvalues of the unitarily equivalent operator $H_{z,\kappa} = -(\nabla + i\kappa)^2+zV$. Using local Lyapounov–Schmidt procedures, which we reformulate as a Schur complement argument, we reduce the problem to the finite-dimensional case described by the effective equation \eqref{eq:det}. 

The challenge is to show that the band spectrum singularities that emerge for small $z$ actually persists for all but discrete values of $z$. Fefferman and Weinstein constructed a vector-valued analytic function $F$, specific to the honeycomb setup, whose zeroes characterize where the multiplicity of an eigenvalue $\mu(z)$ changes. Because the zeroes of $F$ are discrete, this allowed them to show that the multiplicity of $\mu$ is constant away from a discrete set. We instead rely on the theory of holomorphic families of operators of type (A), and specifically a theorem due to \cite{rellich} (Theorem \ref{thm:typeA}). It simplifies the Fefferman--Weinstein procedure by guaranteeing that the eigenvalues of $H_z$ can be represented altogether by analytic functions. Combined this with a Schur complement argument, this implies that the multiplicity of $\mu$ is constant away from a discrete set. 
Furthermore, the corresponding eigenprojector can be extended to an analytic function on $\mathbb{R}$ as well, from which the analyticity of the characteristic polynomial of $M(z,\kappa)$ follows (see Proposition \ref{prop:stronger}).

The proof of Theorem \ref{thm:2}, which we cover in \S \ref{chp:4}, consists of four main steps, again rooted in the work of Fefferman--Weinstein \cite{feffer12}. Given the family of operators $H_z$, where $V$ is periodic and symmetric with respect to a lattice $\Lambda \subset \mathbb{R}^n$ and $K \in \mathbb{R}^n$:
\begin{itemize}
    \item[\textbf{1.}] We first describe the multiplicities of $L^2_K$-eigenvalues of $-\Delta+z V$ for small values of $z$. This combines a perturbation scheme starting from the explicity diagonalizable operator $H_0 = -\Delta$, with a representation-theoretic argument that relies on the specific symmetries of $V$. Multiplicities higher than one corresponds to intersections of dispersion surfaces at $K$ hence to band spectrum singularities near $K$ (see Lemmas \ref{lem:evectors} and \ref{lem:perturbz}). 
    \item[\textbf{2.}] The theory of holomorphic families of operators of type (A), and specifically Proposition \ref{prop:stronger}, ensure that the multiplicities of $L^2_K$-eigenvalues of $H_z$ are actually constant for generic values of $z$ -- in particular, they coincide with those found in Step 1 for small values of $z$. This extends the results of Step 1 to generic values of $z$. 
    \item[\textbf{3.}] We then apply Theorem \ref{thm:1new} to perturb the problem again, this time with respect to $\kappa$, to produce effective equations for the $L^2_{K+\kappa}$-eigenvalues of $H_z$ for $\kappa$ small -- see e.g. \eqref{eq:9d} in the case of three-fold Weyl points.     
    \item[\textbf{4.}] Lastly, we derive expressions for the coefficients of the effective equation in terms of the eigenprojector $\pi(z)$ associated to $\mu(z)$. This allows us to show that these coefficients are non-zero for generic values of $z$, and to therefore describe qualitatively the band spectrum singularities (see Lemma \ref{lem:mkentries}).
\end{itemize}

\subsection{Relation to existing work} 

The goal of this paper is to develop a unified framework for the generic analysis of band spectrum singularities in periodic Schr\"odinger operators. Fefferman--Weinstein \cite{feffer12} showed that honeycomb Schr\"odinger operators generically have Dirac cones in their band spectrum, i.e. conical singularities; multiple related analyses for other models followed. For instance, in two dimensions, \cite{photonic,cassier} generalized the result of \cite{feffer12} to photonic operators; \cite{lieb} showed that Schr\"odinger operators invariant under the Lieb lattice have quadratic degeneracies; \cite{chaban} studied the stability of these degeneracies and showed they split to tilted Dirac cones under parity-breaking perturbations. In three dimensions, \cite{guo} showed that Schr\"odinger operators invariant under the body-centered cubic lattice admit three-fold Weyl points in their band spectrum.

These papers provide a fully detailed analysis for small values of $z$ and $\kappa$ (see Steps 1 and 3 in the proof of Theorem \ref{thm:1new} outlined in Section \ref{sec:1.2}), but later refer to \cite{feffer12} for details about extending their results to generic values of $z$. However, their setup is technically different: for instance, \cite{lieb,chaban} identifies quadratic (instead of linear) singularities and \cite{guo} works with triply (instead of doubly) degenerate eigenvalues. Our paper aims to exempt the above works from providing further details, by proving a general statement about the behavior of eigenvalues and eigenprojectors of lattice-invariant Schr\"odinger operators: Theorem \ref{thm:1new}. 

\subsection{Future projects} The investigation of Dirac cones in honeycomb structures \cite{feffer12} sparked a multitude of mathematical works beside band spectrum singularities: behavior of wavepackets \cite{feffer14}, tight-binding analysis \cite{FeffermanLeeWeinsteinTB}, emergence of edge states \cite{FeffermanLeeWeinstein,DrouotEdgestates,DrouotWeinstein}, propagation of edge states in Dirac systems \cite{BalBeckerDrouot1,BalBeckerDrouot2,DrouotDirac,BalDirac,HXZ}, computation of topological invariants \cite{DrouotBEChoney,Ammari2} and Dirac cones in other setups \cite{berk,Ammari1,WeiJunshanHai}. This showcases the importance of Dirac cones in mathematical physics.

The three-dimensional analogue of Dirac cones are Weyl points (see Definition \ref{def:1}). We believe that they are the only stable type of spectral degeneracies in three dimensions -- see \cite{DrouotWeyl} for an analysis on discrete models. But to the best of our knowledge, one has yet to produce a continuum Schr\"odinger operator with Weyl points. We plan to use the current paper as a stepping stone. Since band spectrum singularities other than Weyl points are believed to be unstable, they should generically split to Weyl points under perturbations. So adding e.g. a parity-breaking term to the Schr\"odinger operators discussed in Theorem \ref{thm:2} should produce Weyl points. This belief is reinforced by a two-dimensional analysis of Chaban--Weinstein \cite{chaban}, who demonstrated that the quadratic degeneracies of Schrodinger operators invariant under square lattices become Dirac cones after adding an odd potential. Constructing Schr\"odinger operators with Weyl points has the potential to spark a number of mathematical investigations, such as wavepacket analysis \cite{feffer14}, study of surface states (the 3D analogues of edge states), and computation of topological invariants \cite{Monaco}.

In \cite{FeffermanLeeWeinsteinTB}, the authors show that high-contrast (large $z$) honeycomb Schr\"odinger operators converge, in an appropriate sense, to their tight binding limit: the Wallace model. As an application, they obtain that the set of values of $z$ so that $H_z$ does not have a Dirac point, is at worst finite. It would be enlightening to perform a similar tight-binding analysis in the case of the three cubic lattices mentioned above, with a tight-binding limit given by the graph Laplacian.

\subsection{Acknowledgement} We gratefully acknowledge support from the National Science Foundation DMS-2054589 and DMS-2439949.

\section{Spectra of analytic families of operators}\label{chp:2}

\subsection{Holomorphic families of type (A)}\label{sec:2.1} 
In order to find interesting dispersion surfaces of the operator $H$ in \eqref{eq:floquet}, we look at 
\begin{equation}
    H_z = -\Delta + z V
\end{equation}
on $L^2_k$  for varying $k$ and generic values of $z \in \mathbb{R}$. For $z$ close to $0$, we can understand the spectrum of $H_z$ using perturbation theory. We rely on analyticity to analyze $H_z$ for $z$ far from $0$, and in particular the theory of {\it holomorphic families of operators of type (A)}, which we briefly review following the work of \cite{kato}.

\begin{definition}\label{def:typeA}
Let $X,Y$ be Banach spaces and let $U \subset \CC$ be an open set. A family of closed operators $T(z):X \rightarrow Y$ for $z \in U$ is said to be {\it holomorphic of type (A)} if
\begin{enumerate}[label = (\arabic*)]
    \item $\dom (T(z)) := \cD$ is independent of $z$;
    \item For all $\phi \in \cD$, $T(z)\phi$ is holomorphic on $U$.
\end{enumerate}
Furthermore, if $\cH = X = Y$ is a Hilbert space and $U \subset \CC$ is symmetric with respect to the real axis, we say the family $T(z)$ is self-adjoint if for all $z \in U$, $T(z)$ is densely defined and $T(z)^* = T(\overline{z})$. 
\end{definition}

An immediate consequence of this definition is that for any $z_0 \in U$ and $\phi \in \cD$, $T(z)\phi$ has a Taylor expansion near $z_0$:
\begin{equation}\label{eq:taylor}
    T(z)\phi = T(z_0)\phi + (z-z_0)T'(z_0)\phi + \frac{1}{2}(z-z_0)^2 T''(z_0)\phi+\cdots.
\end{equation}
Moreover,
\begin{itemize}
    \item This expansion converges in a disk for all $|z-z_0| < r$, independent of $\phi$;
    \item The operators $T^{(n)}(z)$ defined by $T^{(n)}(z)\phi = \frac{d^n (T(z)\phi)}{dz^n}$ are linear.
\end{itemize}

Another important consequence of Definition \ref{def:typeA} is for $z_0,z_1, z_2 \in U$, $|z_1 - z_2|$ sufficiently small, the operator $T(z_1) - T(z_2)$ is {\it relatively bounded} with respect to the resolvent $(T(z_0) - \mu)^{-1}$:

\begin{lemma}\label{lem:relativebound}
Let $T(z):X \rightarrow Y$ be a holomorphic family of type (A) for $z \in U$ with domain $\cD$. Then for any $z_0 \in U$ and $\varepsilon > 0$, there exists $\delta > 0$ such that all $|z_1-z_2|< \delta$, $ \phi \in \cD$, and $\mu \in \rho(T(z_0))$:
\begin{enumerate}
    \item $\|T(z_1)\phi-T(z_2)\phi\| \leq \varepsilon\left (\|\phi\|+\|T(z_0)\phi\|\right)$;
    \item The operator
    \begin{equation}
        (T(z_0) - \mu)^{-1}(T(z_1) - T(z_2))
    \end{equation}
    is bounded as an operator on $\cD$.
\end{enumerate}
\end{lemma}
\begin{proof}
(1) We again turn $\cD$ into a Banach space by introducing the graph norm $\|\phi\|_{T(z_0)} = \|\phi\|+\|T(z_0)\phi\|$, (completeness of $\cD$ follows from the fact that $T(z_0)$ is a closed operator). Then $T(z)$ is closed on $\cD$ with respect to this new norm for all $z \in U$, and so by the closed graph theorem, $T(z)$ is bounded, say by $C(z)$. Choose $r > 0$ such that $\overline{\B_r(z_0)} \subset U$; then for any fixed $\phi \in \cD$, $T(z)\phi$ is analytic and in particular continuous, and so 
\begin{equation}
    \sup_{z \in \overline{\B_r(z_0)}} \|T(z)\phi\|< \infty
\end{equation}
by compactness of $\overline{\B_r(z_0)}$. Thus, we can apply the uniform boundedness principle to the family of bounded operators $\{T(z)\}_{z \in \overline{\B_r(z_0)}}$ to conclude that 
\begin{equation}
    \sup_{z \in \overline{\B_r(z_0)}} \|T(z)\|_{\cL((\cD,\|\cdot\|_{T(z_0)}),Y)} = C < \infty.
\end{equation}

Again let $\phi \in \cD$ be fixed; we can then use Cauchy's integral formula to compute the following bound on $\|T(z_1)\phi-T(z_2)\phi\|$ for $z_1,z_2 $ such that $|z_1-z_0|,|z_2-z_0| \leq r/2$:
\begin{align}
    \|T(z_1)\phi-T(z_2)\phi\| & = \left\| \frac{1}{2\pi i}\oint_{\pd \B_r(z_0)}\frac{T(\zeta)\phi}{\zeta-z_1} - \frac{T(\zeta)\phi}{\zeta-z_2} d\zeta \right\| \\
    & = \left\| \frac{1}{2\pi i}\oint_{\pd \B_r(z_0)}\frac{z_2-z_1}{(\zeta-z_1)(\zeta-z_2)}T(\zeta)\phi\, d\zeta \right\| \\
    & \leq \frac{1}{2\pi}\frac{4}{r^2} C(2\pi r) |z_1-z_2|\,\|\phi\|_{T(z_0)}\\
    & = \frac{4C}{r}|z_1-z_2|\,\|\phi\|_{T(z_0)}.
\end{align}

Therefore, if we let 
\begin{equation}
    \delta = \min \left\{\frac{r}{2},\frac{\varepsilon r}{4C}\right\},
\end{equation}
then for all $\phi \in \cD$ and $z_1,z_2$ such that $|z_1-z_2| < \delta$,
\begin{equation}
    \|T(z_1)\phi-T(z_2)\phi\| \leq \varepsilon\|\phi\|_{T(z_0)} = \varepsilon\left (\|\phi\|+\|T(z_0)\phi\|\right).
\end{equation}

(2) Let $\mu \in \rho(T(z_0))$; then $(T(z_0) - \mu)^{-1}$ is a bounded linear operator from $Y$ to $\cD$, and consequently $(T(z_0) - \mu)^{-1}T(z)$ is a holomorphic family of type (A) for $z \in U$ with domain $\cD$. Thus, by repeating the arguments in part (1), we deduce that 
\begin{equation}
    \sup_{z \in \overline{\B_r(z_0)}} \|(T(z_0) - \mu)^{-1}T(z)\|_{\cL((\cD,\|\cdot\|_{T(z_0)}),Y)} = C\|(T(z_0) - \mu)^{-1}\|.
\end{equation}
Consequently, by shrinking $\delta$ if necessary so that 
\begin{equation}
    \delta \leq \frac{\varepsilon r}{4C\|(T(z_0) - \mu)^{-1}\|},
\end{equation}
we conclude that for $\phi \in \cD$ and $z_1,z_2$ such that $|z_1-z_2| < \delta$,
\begin{equation}\label{eq:relativebound}
    \|(T(z_0) - \mu)^{-1}(T(z_1)-T(z_2))\phi\| \leq  \varepsilon\left (\|\phi\|+\|(T(z_0) - \mu)^{-1}T(z_0)\phi\|\right).
\end{equation}
However, also observe that
\begin{align}
    \|(T(z_0) - \mu)^{-1}T(z_0)\phi\| & = \|(T(z_0) - \mu)^{-1}(T(z_0) - \mu) + \mu(T(z_0) - \mu)^{-1}\phi\| \\
    & = \|\phi+ \mu(T(z_0) - \mu)^{-1}\phi\| \\
    & \leq \left(1+ \mu\|(T(z_0) - \mu)^{-1}\|\right)\|\phi\|.
\end{align}
Plugging this back into \eqref{eq:relativebound}, we get that
\begin{equation}
    \|(T(z_0) - \mu)^{-1}(T(z_1)-T(z_2))\phi\| \leq \varepsilon \left(2+ \mu\|(T(z_0) - \mu)^{-1}\|\right)\|\phi\|.
\end{equation}
Therefore, for $z_1,z_2$ such that $|z_1-z_2| < \delta$, $(T(z_0) - \mu)^{-1}(T(z_1) - T(z_2))$ is bounded as an operator on $\cD$.
\end{proof}

We shall see in \S\ref{chp:3} that the family of operators $H_z$ defined in \eqref{eq:Hz} is a self-adjoint holomorphic family of type (A) on $L^2_k$. For our purposes, one of the most important results for such families is the following theorem due to \cite{rellich}.

\begin{theorem}\label{thm:typeA}
Let $\cH$ be a Hilbert space and let $T(z)$ be a self-adjoint holomorphic family of type (A) on $\cH$, defined on a neighborhood $U$ of an interval $I_0$ of the real axis. Furthermore, assume that $T(z)$ has a compact resolvent for $z \in U$. Then, there exists a  sequence of scalar-valued functions $(\mu_n(z))_{n \in \NN}$ and a sequence of vector-valued functions $(\phi_n(z))_{n \in \NN}$, all analytic on $I_0$, such that for $z \in I_0$, $(\phi_n(z))_{n \in \NN}$ form a complete orthonormal basis of eigenvectors of $T(z)$, with corresponding eigenvalues $(\mu_n(z))_{n \in \NN}$.
\end{theorem}

\subsection{Variation of eigenvalues}\label{sec:2.3}
In addition to the above tools, the proof of our main theorem requires some techniques from the theory of variation of eigenvalues. In this section, we will restrict our attention to families of operators $T(z)$ satisfying the hypotheses of Theorem \ref{thm:typeA}; namely that $T(z)$ is a self-adjoint holomorphic family of type (A) with compact resolvent, defined on a neighborhood $U$ of an interval $I_0$ of the real axis. We shall also let $\phi$ and $\mu$ be analytic functions on $U$ such that $\phi(z)$ and $\mu(z)$ are a unit length eigenvector and eigenvalue, respectively, of $T(z)$ for all $z \in U$, whose existence is guaranteed by Theorem \ref{thm:typeA}.

\begin{lemma}\label{lem:eigenproj}
Let $z_0 \in U$; there exist $\varepsilon, \delta > 0$ such that:
\begin{enumerate}[label=(\arabic*)]
    \item $\mu_0: = \mu(z_0)$ is the only eigenvalue of $T(z_0)$ in $\B_\varepsilon(\mu_0)$,
    \item for every $z \in \B_\delta(z_0)$, $T(z)$ has no eigenvalue on $\pd \B_\varepsilon(\mu_0)$,
    \item for every $z \in \B_\delta(z_0)$, the operator
    \begin{equation}\label{eq:eigenproj}
        \pi(z)  := -\frac{1}{2\pi i}\oint_{\pd \B_\varepsilon(\mu_0)} (T(z)-\lambda)^{-1}d\lambda
    \end{equation}
    is an analytic family of projectors, whose rank is independent of $z$.
\end{enumerate}
\end{lemma}

The operator $\pi(z)$ defined in \eqref{eq:eigenproj} is typically called the {\it spectral (or Riesz) projector} corresponding to $\pd \B_\varepsilon(\mu_0)$. By evaluating this operator on an eigenvector corresponding to an eigenvalue $\widetilde{\mu}$ of $T(z)$ contained in $\B_\varepsilon(\mu_0)$ and using Cauchy's integral formula, one can check that this operator restricts to the identity on the corresponding eigenspace, and in particular the image of this operator contains all eigenspaces corresponding to eigenvalues $\widetilde{\mu}$ of $T(z)$ contained in $\B_\varepsilon(\mu_0)$.

\begin{proof}[Proof of Lemma \ref{lem:eigenproj}]
(1)+(2): Since the spectrum of $T(z_0)$ is discrete, there exists a $\varepsilon > 0$ such that $\mu_0 : = \mu(z_0)$ is the only eigenvalue of $T(z_0)$ contained in $\overline{\B_\varepsilon(\mu_0)}$ (so that (1) automatically holds). To prove (2), note that for any $\lambda \in \pd \B_\varepsilon(\mu_0)$, we can write
\begin{align}\label{eq:neuman}
    T(z) - \lambda & =  (T(z_0) - \lambda)\cdot(I + K_\lambda(z)), \quad \text{ where }\\
    K_\lambda(z) & = (T(z_0) - \lambda)^{-1}(T(z)-T(z_0)).
\end{align}
Since $\pd \B_\varepsilon(\mu_0)$ is compact and $(T(z_0) - \lambda)^{-1}$ is analytic, and thus continuous, in $\lambda$ for all $\lambda \in \pd \B_\varepsilon(\mu_0)$, there exists $C > 0$ such that 
\begin{equation}
    \lambda\|(T(z_0) - \lambda)^{-1}\| \leq C
\end{equation}
for all such $\lambda$. Therefore, if we let
\begin{equation}
    \varepsilon_0 = \frac{1}{2(2+C)},
\end{equation}
then by Lemma \ref{lem:relativebound} (and its proof), there exists $\delta > 0$ such that for $|z-z_0| < \delta$ and for all $\lambda \in \pd \B_\varepsilon(\mu_0)$,
\begin{equation}
    \|K_\lambda(z)\| \leq \varepsilon_0(2 + \lambda\|(T(z_0) - \lambda)^{-1}\|) < 1.
\end{equation}
Consequently, $I + K_\lambda(z)$ is invertible, and since $T(z_0) -\lambda$ is also invertible for all $\lambda \in \pd \B_\varepsilon(\mu_0)$, we deduce that $T(z) - \lambda$ is invertible as well for all such $\lambda$ and $z \in \B_\delta(Z_0)$, thus proving the claim.

(3) This justifies that the operator $\pi(z)$ in \eqref{eq:eigenproj} is well-defined. In addition, it is analytic since its integrand is analytic for all $z \in \B_\delta(z_0)$. To see that it is a projector for all such $z$, pick $0 < \varepsilon_0<\varepsilon$ such that $\B_{\varepsilon_0}(\mu_0)$ contains the same eigenvalues of $T(z)$ as $\B_\varepsilon(\mu_0)$ (which exists since the resolvent set of $T(z)$ is open); then by the residue theorem, $\pi(z)$ is also equal to the integral in \eqref{eq:eigenproj}, but with the $\pd \B_\varepsilon(\mu_0)$ replaced with $\pd \B_{\varepsilon_0}(\mu_0)$. Let $\cC = \pd \B_\varepsilon(\mu_0)$ and let $\cC_0 = \pd \B_{\varepsilon_0}(\mu_0)$; then, by the first resolvent identity,
\begin{align}
    \pi(z)^2 & = \frac{1}{(2\pi i)^2}\oint_{\cC_0} (T(z)-\mu)^{-1}d\mu \oint_{\cC} (T(z)-\lambda)^{-1}d\lambda \\
    & = \frac{1}{(2\pi i)^2}\oint_{\cC_0} \oint_{\cC} \frac{(T(z)-\mu)^{-1}-(T(z)-\lambda)^{-1}}{\mu-\lambda} d\lambda d\mu \\
    & = \frac{1}{(2\pi i)^2}\left(\oint_{\cC_0} (T(z)-\mu)^{-1}\oint_{\cC} \frac{1}{\mu-\lambda} d\lambda d\mu - \oint_{\cC}(T(z)-\lambda)^{-1} \oint_{\cC_0} \frac{1}{\mu-\lambda} d\mu d\lambda \right) \\
    & = \frac{1}{(2\pi i)^2}\left(\oint_{\cC_0} (T(z)-\mu)^{-1}(-2\pi i) d\mu - \oint_{\cC}(T(z)-\lambda)^{-1}(0) d\lambda \right) \\
    & = -\frac{1}{2\pi i}\oint_{\cC_0} (T(z)-\mu)^{-1} d\mu  = \pi(z).
\end{align}

Lastly, to show that the rank of $\pi(z)$ is independent of $z$, we use a lemma due to Kato \cite[Lemma I.4.10]{kato}: if $\pi_1$, $\pi_2$ are two projectors such that $\|\pi_1 - \pi_2\| < 1$, then $\pi_1$ and $\pi_2$ have the same (potentially infinite) rank.\footnote{In Kato’s book this property is shown for projectors on a finite-dimensional vector space. One can check that the proof applies to infinite-dimensional vector spaces as well.} It follows that the set $\{z \in \B_\delta(z_0)\;:\; \rank \pi(z) = \pi(z_0)\}$
is non-empty, open, and closed, and thus equal to $\B_\delta(z_0)$. 
\end{proof}

We will later use the following corollary:

\begin{corollary}\label{cor:simple} Let $z_0 \in I_0$ such that $\mu_0 := \mu(z_0)$ is an eigenvalue of $T(z_0)$ of multiplicity $1$ and let $\varepsilon, \delta$ be the quantities produced by Lemma \ref{lem:eigenproj}. For every $z \in  (z_0\pm\delta)$, $\mu(z)$ is the only eigenvalue of $T(z)$ in $[\mu_0\pm \varepsilon]$ and satisfies
\begin{equation}\label{eq:firstorder}
    \mu(z) = \mu_0 + (z-z_0) \cdot \langle \phi(z_0),T'(z_0) \phi(z_0) \rangle + O(z-z_0)^2.
\end{equation}
\end{corollary}

\begin{proof} By Theorem \ref{thm:typeA}, for $z \in I_0$ we can write $T(z)$ as:
\begin{equation}
    T(z) = \sum_{n=1}^\infty \mu_n(z) \langle \cdot ,\phi_n(z)\rangle \phi_n(z).
\end{equation}

In particular, by the residue theorem, for $z \in (z_0 \pm \delta)$ the projector $\pi(z)$ from \eqref{eq:eigenproj} is given by:
\begin{equation}\label{eq:1b}
    \pi(z) = \sum_{n: \, |\mu_n(z) - \mu_0| < \varepsilon} \langle \cdot ,\phi_n(z)\rangle \phi_n(z).
\end{equation}
We note that $\pi(z_0)$ has rank one, because $\mu_0$ is the only eigenvalue of $T(z_0)$ in $\B_\varepsilon(\mu_0)$. Therefore, $\pi(z)$ has rank one as well. We deduce from \eqref{eq:1b} that for every $z \in  (z_0\pm\delta)$, $\mu(z)$ is the only eigenvalue of $T(z)$ in $\B_\varepsilon(\mu_0)$, and $\pi(z)$ is its corresponding eigenprojector. Moreover, since $T(z)$ is a self-adjoint family, we deduce that, in fact, $\mu(z) \in (\mu_0\pm \varepsilon)$.

To see that $\mu(z)$ has the form \eqref{eq:firstorder}, first note that $T(z_0)$ is self-adjoint, and that \newline $\Re \langle \phi'(z),\phi(z)\rangle = 0$ since $\|\phi(z)\| = 1$ for all $z \in U$. As a result,
\begin{align}
    \mu'(z_0) & = \frac{d}{dz} \langle \phi(z),T(z)\phi(z) \rangle\big |_{z=z_0} \\
    & =  \langle \phi'(z_0),T(z_0)\phi(z_0) \rangle + \langle \phi(z_0),T'(z_0)\phi(z_0) + T(z_0)\phi'(z_0)\rangle \\
    & = \phi(z_0),T'(z_0)\phi(z_0) \rangle + 2 \mu(z_0) \Re\langle \phi'(z_0),\phi(z_0)\rangle \\
    & = \phi(z_0),T'(z_0)\phi(z_0) \rangle.
\end{align}
Using a Taylor expansion of $\mu$, we conclude that:
\begin{equation}
    \mu(z) = \mu_0 + (z-z_0) \langle \phi_0,T'(z_0)\phi_0 \rangle + O(z-z_0)^2.
\end{equation}
This completes the proof. \end{proof}

In addition to the analyticity of the eigenvalue $\mu(z)$ of $T(z)$, we shall also need to track the multiplicity of such eigenvalues to ensure that, for generic $z$, the number of dispersion surfaces involved in a given band spectrum singularities remains constant. The following proposition, which is a direct consequence of Theorem \ref{thm:typeA}, addresses this.

\begin{proposition}\label{prop:stronger}
There exists a discrete set $D \subset I_0$ such that, as an eigenvalue of $T(z)$, $\mu(z)$ has constant multiplicity for all $z \in I_0 \setminus D$. Furthermore, there exists an analytic function $\pi(z):I_0 \rightarrow \cL(\cH)$ such that $\pi(z)$ is an orthogonal projection of constant rank for all $z \in I_0$, and $\pi(z)$ is the eigenprojector associated to $\mu(z)$ for all $z \in I_0 \setminus D$.
\end{proposition}
\begin{proof}
First we prove Proposition \ref{prop:stronger} when the Hilbert space $\cH$ is finite-dimensional. By picking some fixed basis for $\cH$, $T(z)$ can be represented by a family of matrices, which we denote by $M(z)$, the entries of which will also be analytic by Definition \ref{def:typeA} (2). Let $(\phi_j(z))_{j=1}^n$, $(\mu_j(z))_{j=1}^n$ be the vector-valued and scalar-valued functions, respectively, whose existence is guaranteed by Theorem \ref{thm:typeA}. After reindexing these functions if necessary, we may assume that $\mu = \mu_1$. 

For $j=2,\ldots,n$, define
\begin{equation}\label{eq:Djs}
    D_j = \{z \in I_0\;:\; \mu_j(z) - \mu_1(z) = 0\}.
\end{equation}
Since $\mu_j-\mu$ is analytic, $D_j$ must either be discrete or equal to $I_0$ by the identity theorem. It follows that the set 
\begin{equation}\label{eq:Djs2}
    D = \bigcup_{j\;: D_j \not = I_0} D_j
\end{equation}
is also discrete as a finite union of discrete sets. We additionally define a function $\pi:I_0 \rightarrow \cL(\cH)$ by
\begin{equation}\label{eq:pi(z)}
    \pi(z)\phi = \sum_{j\;:D_j = I_0} \langle \phi,\phi_j(z)\rangle \phi_j(z).
\end{equation}
By construction, $\pi(z)$ is both analytic for all $z \in I_0$ and equal to the eigenprojector associated to $\mu(z)$ for $z \in I_0 \setminus D$. Furthermore, it is an orthogonal projector for all $z \in I_0$ by virtue of $(\phi_j(z))_{j=1}^n$ forming a complete orthonormal basis of $\cH$ for all $z\in I_0$. This also implies that $\pi(z)$ has constant rank on $I_0$, and therefore $\mu(z)$ must have constant multiplicity for all $z \in I_0 \setminus D$.

We now prove Proposition \ref{prop:stronger} when the Hilbert space $\cH$ is potentially infinite-dimensional by reducing to the finite-dimensional case. Just as before, let $(\phi_j(z))_{j \in \NN}$, $(\mu_j(z))_{j\in \NN}$ be the vector-valued and scalar-valued functions, respectively, whose existence is guaranteed by Theorem \ref{thm:typeA} (again potentially reindexing so that $\mu = \mu_1$), and let $D_j$, $D$, and $\pi$ be defined as in \eqref{eq:Djs}, \eqref{eq:Djs2}, and \eqref{eq:pi(z)}, respectively. 

To see that $\pi$ is well-defined, and in particular that the sum in its definition is finite, note that, although we can no longer assume $D$ is discrete, $D$ is still countable as a countable union of discrete sets. As a result, there exists some $z_0 \in I_0 \setminus D$; by construction, $\pi(z_0)$ is then the eigenprojector corresponding to $\mu(z_0)$, and since the spectrum of $T(z_0)$ is discrete, $\pi(z_0)$ must have finite rank. Let $m = \rank \pi(z_0)$; it follows that the sum in \eqref{eq:pi(z)} has $m$ terms, and so $\pi$ is well-defined as a function on $I_0$. Furthermore, $\pi$ is again an orthogonal projector of constant rank $m$, and for $j \in \NN$ such that $D_j \not = I_0$, the functions $\phi_j(z)$ are eigenvectors of $T(z)$ corresponding to the eigenvalue $\mu(z)$ for all $z \in I_0$. In particular, this tells us that:
\begin{itemize}
    \item $\pi(z)$ is analytic for all $z \in I_0$;
    \item $\pi(z)$ is the eigenprojector corresponding to $\mu(z)$ for all $z \in I_0 \setminus D$;
    \item $\mu(z)$ has multiplicity $m$ for all $z \in I_0\setminus D$ (and multiplicity $\geq m$ for all $z \in I_0$).
\end{itemize}
It thus remains to show that the set $D$ is discrete.

Henceforth, let $\mult_{T(z)}(\lambda)$ denote the multiplicity of $\lambda$ (possibly zero) as an eigenvalue of $T(z)$, let $z_0 \in D$, and let $\mu_0 : = \mu(z_0)$. Then $m_0 := \mult_{T(z_0)}(\mu(z_0)) \geq m$. In addition, if we apply Lemma \ref{lem:eigenproj} to $z_0, \mu_0$ and let $\widetilde{\pi}(z)$ be the operator defined in \eqref{eq:eigenproj}, then for some $\varepsilon,\delta >0$ and for $z \in \B_{\delta}(z_0))$, $\widetilde{\pi}(z)$ is the spectral projector corresponding to eigenvalues contained in $\B_\varepsilon(\mu_0)$. 

Let $\cE(z) =  \widetilde{\pi}(z)(\cH)$; then $\cE(z)$ is a finite-dimensional vector space of dimension independent of $z$ since $\rank \widetilde{\pi}(z) = \rank \widetilde{\pi}(z_0)$ by  Lemma \ref{lem:eigenproj} (3) and $\rank \widetilde{\pi}(z_0) = \mult_{T(z_0)}(\mu_0) = m_0 < \infty$ due to the spectrum of $T(z_0)$ being discrete. Since $T(z)$ is assumed to be acting on a Hilbert space $\cH$, we can decompose $T(z)$ with respect to $ \cE(z) \oplus \cE(z)^\perp$:
\begin{equation}\label{eq:decomp}
    T(z) = \begin{pmatrix} T_{11}(z) & T_{12}(z) \\
    T_{21}(z) & T_{22}(z)
    \end{pmatrix}.
\end{equation}
For $z \in (z_0 \pm \delta)$, $T_{12}(z) = T_{21}(z) = 0$ since $T(z)$ is self-adjoint for $z \in I_0$. Moreover, these operators are analytic (as they can be expressed as compositions of $T(z)$, $\widetilde{\pi}(z)$, and the orthogonal complement of $\widetilde{\pi}(z)$), and thus the identity theorem tells us they must be identically zero on $\B_\delta(z_0)$.

In addition, note that by construction, $T_{22}(z_0) = T(z_0)|_{\cE(z_0)^\perp}$ has no eigenvalues in $\B_\varepsilon(\mu_0)$, and since $T(z_0)$ is self-adjoint, this implies $T(z_0) - \mu_0$ is invertible and its norm is bounded by $1/\delta$. By writing
\begin{align}
    T_{22}(z) - \lambda & =  (T_{22}(z_0) - \mu_0)\cdot(I + K_\lambda(z)), \quad \text{ where }\\
    K_\lambda(z) & = (T_{22}(z_0) - \mu_0)^{-1}(T_{22}(z) -T_{22}(z_0)+\lambda - \mu_0),
\end{align}
we obtain, by the same argument as in the proof of Lemma \ref{lem:eigenproj},  that $T_{22}(z) - \lambda$ is invertible for $|z - z_0| < \delta$ (after shrinking $\delta$ if necessary) and $\lambda \in \B_\varepsilon(\mu_0)$.

Let $\psi_1,\ldots,\psi_{m_0}$ be a basis of $\cE(z_0)$. After shrinking $\delta$ again if necessary, the set
 \begin{equation}
    \{\psi_j(z) := \widetilde{\pi}(z)\psi_j\}_{j=1}^{m_0}
\end{equation}
forms a basis for $\cE(z)$. Let $M(z)$ be the matrix of $T_{11}(z)$ with respect to this basis. Then $M(z)$ is Hermitian for $z \in (z_0 \pm \delta)$, and its entries are given by 
\begin{equation}
    (M(z))_{ij} = \frac{\langle T(z)\psi_i(z),\psi_j(z)\rangle}{\|\psi_j(z)\|^2},
\end{equation}
from which it follows that $M(z)$, viewed as a linear operator on $\CC^{m_0}$, is a self-adjoint holomorphic family of type (A) for $z \in \B_\delta(z_0)$.

Since $T_{12}(z) = T_{21}(z) = 0$ and $T_{22}(z) - \lambda$ is invertible for $|z - z_0| < \delta$ and $\lambda \in \B_\varepsilon(\mu_0)$, the decomposition in \eqref{eq:decomp} implies the following sequence of equivalences for $z \in (z_0 \pm \delta)$:
\begin{align}
    \lambda \in \B_\varepsilon(\mu_0) \text{ is an eigenvalue of }T(z) &  \Leftrightarrow \lambda \in \B_\varepsilon(\mu_0) \text{ such that }T(z) - \lambda \text{ is singular} \\
    & \Leftrightarrow T_{11}(z) - \lambda \text{ is singular} \\
    & \Leftrightarrow M(z) - \lambda \text{ is singular} \\
    & \Leftrightarrow \lambda \text{ is an eigenvalue of }M(z).
\end{align}
Moreover, $\mu(z)$ is an eigenvalue of $M(z)$ for all $z \in (z_0 \pm \delta)$ since $\mu((z\pm z_0)) \subset  \B_\varepsilon(\mu_0)$ by Lemma \ref{lem:eigenproj}. As a result, by applying this proposition to the finite-dimensional family of operators $M(z)$, we deduce that $\mu(z)$ has constant multiplicity on a punctured interval of $z_0$. Since the multiplicity of $\mu(z)$ as an eigenvalue of $M(z)$ is equal to its multiplicity as an eigenvalue of $T(z)$, we conclude that $\mu(z)$ has multiplicity $m$ as an eigenvalue of $T(z)$ in a punctured neighborhood of $z_0$. Since $z_0 \in D$ was arbitrary, this shows that $D$ is in fact discrete.
\end{proof}

\section{Dispersion Surfaces of Schr{\"o}dinger Operators: General Theory}\label{chp:3} 

The rest of this paper focuses on the eigenvalue problem \eqref{eq:9b}. Specifically, using the theory of holomorphic families of type (A), we develop a framework for analyzing the dispersion surfaces of Schr{\"o}dinger operators $-\Delta + V$ for generic potentials $V$ invariant under a lattice $\Lambda$, i.e. periodic with respect to $\Lambda$ and symmetric with respect to the point group of $\Lambda$.

This section concentrates on this general set-up. We first review Floquet--Bloch theory and define lattice-invariant potentials. We then state and prove perturbative lemmas on Floquet--Bloch eigenvalues of $H_z = -\Delta+zV$, in the process proving our first main result, Theorem \ref{thm:1new}. Brought together, these results outline our strategy to describe the generic structure of the dispersion surfaces of invariant Schr{\"o}dinger operators. In the next section we apply this framework to analyze the dispersion surfaces of Schr{\"o}dinger operators with potentials invariant under cubic lattices.

\subsection{Floquet--Bloch Theory}\label{sec:3.1}
We begin with a review of lattices and Floquet--Bloch theory. Given a basis $v_1,\ldots,v_n$ of $\mathbb{R}^n$, the {\it lattice} $\Lambda$ generated by $v_1,\ldots,v_n$ is the set $\Lambda = \Z v_1\oplus \cdots \oplus \Z v_n$. Given $k \in \mathbb{R}^n$, the space of $k$-quasiperiodic functions with respect to $\Lambda$ is
\begin{equation}
    L^2_k = \{f \in L^2_{\text{loc}}\;:\; f(x+v) = e^{ik\cdot v}f(x) \; \forall v \in \Lambda\}.
\end{equation}
In this context, we refer to $k$ as the {\it quasi-momentum} of functions $f \in L^2_k$. In addition, observe that the space of $\Lambda$-periodic functions is simply $L^2_0$, and $f \in L^2_k$ if and only if $e^{-ik\cdot x}f \in L^2_0$. The correspondence between $L^2_0$ and $L^2_k$ then induces an inner product on $L^2_k$ given by:
\begin{equation}
    \langle f,g \rangle_{L^2_k} = \frac{1}{|\mathbb{R}^n/\Lambda|}\int_{\mathbb{R}^n/\Lambda} \overline{f(x)}g(x)dx,
\end{equation}
where $\mathbb{R}^n / \Lambda$ is a fundamental cell for $\Lambda$. We similarly define Sobolev spaces $H^s_k$, $s \in \N$ by
\begin{equation}
    H^s_k = \{f\in L^2_k\;:\; \pd^\alpha f\in L^2_k \; \forall |\alpha|\leq s\}.
\end{equation}

Lastly, we define the {\it dual lattice} $\Lambda^*$ (also often referred to as the {\it reciprocal lattice}) as $\Lambda^* = \Z k_1 \oplus \cdots \Z k_n$, where $k_1,\ldots, k_n$ satisfy the relation $k_i\cdot v_j = 2\pi \delta_{ij}$. We then refer to $k_1,\ldots, k_n$ as the {\it dual} (or {\it reciprocal})  {\it basis}.

We consider the Schr{\" o}dinger operator $H = - \Delta + V$, where $V$ is smooth and periodic with respect to $\Lambda$. The {\it Floquet--Bloch eigenvalue problem} at quasi-momentum $k \in \mathbb{R}^n$ is
\begin{equation}\label{eq:floquet}
\begin{gathered}
    H\phi(x;k) = \mu(k) \phi(x;k), \quad x \in \mathbb{R}^n,\\
    \phi(x+v;k) = e^{ik \cdot v}\phi(x;k), \quad v \in \Lambda.
    \end{gathered}
\end{equation}
A $L^2_k$-solution $\phi$ to the above problem is called a {\it Floquet--Bloch state}. The operator $H$ is a self-adjoint unbounded operator on $L^2_k$ (respectively $L^2$) with domain $H^2_k$ (respectively $H^2$). By elliptic regularity, the operator $H$ on $L^2_k$ has a compact resolvent, and so its spectrum is discrete; the collection of its eigenvalues, seen as functions of $k$, are called the {\it dispersion surfaces} of $H$. 

Since the problem \eqref{eq:floquet} is invariant under the change $k \mapsto k + k'$ for $k' \in \Lambda^*$, we can restrict our attention to $k$ varying over the {\it Brillouin zone} $\cB$: the set of points $k \in \mathbb{R}^n$ which are closer to the origin than to any other point of $\Lambda^*$. Moreover, we can recover the $L^2$-spectrum of $H$ from the $L^2_k$-spectra  for $k \in \cB$ \cite{reed}:
\begin{equation}
    \sigma_{L^2}(H) = \bigcup_{k \in \cB}\sigma_{L^2_k}(H).
\end{equation}

\subsection{Invariant Potentials}\label{sec:3.2}
In this section, we fix a lattice $\Lambda$ with basis $v_1,\ldots,v_n$ and reciprocal basis $k_1,\ldots,k_n$. Let $G$ denote the {\it point group} of the lattice $\Lambda$, namely the subgroup of its isometry group which keeps the origin fixed. Observe that $G$ is necessarily finite: every element $g \in G$ must necessarily send the basis $v_1,\ldots,v_n$ to another basis of $\mathbb{R}^n$ consisting of vectors in $\Lambda$, and since $g$ is an isometry, we must have that $\|gv_j\| = \|v_j\|$ for $j=1,\ldots,n$, which implies there are only finitely many lattice vectors to which $g$ can send each basis element. 

The group $G$ acts isometrically on scalar-valued functions:
\begin{equation}
    g_*f(x) := f(g^\top x).
\end{equation}
We will later need an induced action of a subgroup $G_0$ of $G$ on $L^2_k$ for some quasi-momentum $k$. However, in order for this action to be well-defined, we need $G_0$ to satisfy an additional criterion. 

\begin{definition}
We say $g \in G$ is $k$-{\it invariant} if
\begin{equation}
    gk \in k+\Lambda^*.
\end{equation}
Analogously, we say a subgroup $G_0$ of $G$ is $k$-invariant if $g$ is $k$-invariant for all $g \in G_0$.
\end{definition}

To see that $k$-invariant subgroups give well-defined actions, note that if $G_0$ is such a subgroup and $g \in G_0$, then by definition there exists $k' \in \Lambda^*$ such that $gk = k+k'$. Then for all $v \in \Lambda$, $g^\top v \in \Lambda$ as well by definition of $G$, and as a result
\begin{align}
    g_*\psi(x+v)&  = \psi(g^\top x+g^\top v) = e^{ik\cdot g^\top v}\psi(g^\top x)
    = e^{igk\cdot v}g_*\psi(x) \\
    & =e^{i(k+k')\cdot v}g_*\psi(x)  = e^{ik\cdot v}g_*\psi(x).
\end{align}
In particular, this shows that $k$-invariant group elements map $L^2_k$ to itself.

We now define potentials invariant with respect to $\Lambda$.

\begin{definition}\label{def:invariant}
Let $\Lambda$ be a lattice with point group $G$. We say that $V\in C^\infty(\mathbb{R}^n,\mathbb{R})$ is $\Lambda$-{\it invariant} if:
\begin{enumerate}[label = \arabic*)]
    \item $V$ is $\Lambda$-periodic, i.e. $V(x+v) = V(x)$ for all $x \in \mathbb{R}^2$ and $v \in \Lambda$,
    \item $V$ is $G$-invariant, i.e. $g_* V = V$ for all $g \in G$.
\end{enumerate}
When the lattice $\Lambda$ is clear from the context, we will omit it and simply refer to $V$ as an {\it invariant potential}.
\end{definition}

When $V$ is an invariant potential, the fact that $V$ is $\Lambda$-periodic enables us to expand $V$ as a Fourier series with coefficients $\{V_m\}_{m \in \Z^n}$:
\begin{align}\label{eq:Fouriercoef}
    V(x) & = \sum_{m \in \Z^n}V_me^{i(m_1k_1+\cdots +m_nk_n)\cdot x} \\
    V_m & = \left \langle e^{i(m_1k_1+\cdots +m_nk_n)\cdot x},V \right\rangle.
\end{align}
For simplicity of notation, if $k \in \Lambda^*$ so that $k = m_1k_1 + \cdots + m_nk_n$ for some $m \in \Z^n$, we shall also denote $V_m$ by $V_{k}$. If we then view these coefficients as a function on $\Lambda$, they are invariant under an induced action of $G$: 
\begin{equation}\label{eq:Fourierinvar}
    g_*V_{k}  = V_{g^\top k} = \left \langle e^{ig^\top k\cdot x},V \right\rangle = \left \langle e^{ik\cdot x},g_*V \right\rangle = \left \langle e^{ik\cdot x},V \right\rangle=  V_{k}.
\end{equation}

An example of invariant potentials that has been studied extensively is honeycomb lattice potentials: potentials invariant under $2\pi/3$-rotations and parity and periodic with respect to the equilateral lattice. For later reference, we now describe two properties of invariant potentials which naturally extend properties of honeycomb lattice potentials.

First, observe that if $V$ is a $\Lambda$-invariant potential and $O$ is an orthogonal transformation, then $V \circ O^*$ is an $O\Lambda$-invariant potential. An immediate consequence of this is that the spectral properties of $H_z$ on $L^2_k$ are the same as those of $H_z$ on $L^2_{gk}$ for all $g \in G$. Together with the $\Lambda^*$-periodicity of the Floquet-eigenvalue problem \eqref{eq:floquet}, this implies that the dispersion surfaces of $H$ near a quasi-momenta $k \in \mathbb{R}^n$ are determined locally by those near $gk$. Consequently, it suffices to consider quasi-momenta whose orbits under $G$ are distinct.

Second, every lattice $\Lambda$ is necessarily invariant under the negative of the identity, which implies that $-I \in G$. Therefore, by $G$-invariance, every invariant potential $V$ is necessarily even. Together with the assumption that $V$ is real, this implies that if $(\phi(x;k),\mu)$ is an eigenpair of the Floquet--Bloch problem \eqref{eq:floquet} with quasi-momentum $k$, then so too is $(\overline{\phi(-x;k)},\mu)$.

\subsection{Decomposing \texorpdfstring{$L^2_{K}$}{L2K} via a \texorpdfstring{$K$}{K}-Invariant Subgroup}\label{sec:3.3}

Fix some $K \in \mathbb{R}^n$; then $\mu_0 = \|K\|^2$ is an $L^2_K$-eigenvalue of $-\Delta$. We define a set $[K]$ as follows:  
\begin{equation}\label{eq:equivclass}
    [K]: = \{k \in K+\Lambda^*\;:\; \|k\|=\|K\|\}.
\end{equation}
Then by Corollary \ref{cor:brillouinvert},
\begin{equation}
    m_{-\Delta}(\mu_0) = \Big|[K]\Big|.
\end{equation}

For the rest of this section, we make the following assumption on $K$:\\

\noindent {\bf Assumption \setword{1}{assump:3}:} There exists an abelian subgroup $G_0$ of $G$ such that $G_0K = [K]$ and $|G_0| = |G_0K|$. \\

Although this might appear at first glance to be a restrictive assumption, we will see in \S\ref{chp:4} that in many applications, such a subgroup exists. The reason this assumption is helpful is that by construction of $[K]$, $G_0$ is necessarily $K$-invariant, and thus has a well-defined action on $L^2_K$. In addition, by our assumption that $V$ is $G$-invariant and the fact that $g_*$ is the pushforward by an orthogonal matrix for every $g \in G_0$, $H$ commutes with the action of $G_0$ on $L^2_K$. We can therefore reduce the spectral problem for $H$ on $L^2_K$ to spectral problems on the invariant subspaces of $G_0$. 

Before we perform this reduction, however, we introduce some notation. Let $g_1,\ldots,g_\ell$ denote a minimal system of generators of $G_0$, with respective orders $n_1,\ldots,n_\ell$. Since $G_0$ is assumed to be abelian, it follows that $G_0 \cong \bigoplus_{j=1}^\ell \Z_{n_j}$. In addition, if $g \in G_0$ is of order $N$, then $\sigma_{L^2_{K}}(g_*)$, the spectrum of $g_*$ viewed as an operator on $L^2_K$, is contained in the $N$-th roots of unity $U_N$ (and in fact, we will see in Lemma \ref{lem:evectors} that $\sigma_{L^2_{K}}(g_*)=U_N$). This follows first from the fact that $g_*$ has finite order, and consequently has pure point spectrum, and if $\omega$ is an eigenvalue of $g_*$, then $g^N = e$ implies $\omega^N = 1$, and so $\omega \in U_N$. With this in mind, we define:
\begin{equation}\label{eq:JU}
    \JJ := \prod_{j=1}^\ell \{0,\ldots,n_j-1\} \quad \text{ and }\quad \UU  := \prod_{j=1}^\ell U_{n_j},
\end{equation}
so that $G_0  = \{g^j\;:\;j \in \JJ\}$, where we are using the multi-index notation $g^j = g_1^{j_1}\cdots g_\ell^{j_\ell}$.

Again using the fact that $G_0$ is abelian, we can then simultaneously diagonalize the operators $(g_j)_*$, which leads us to the following decomposition of $L^2_K$:
\begin{equation}
    L^2_K = \bigoplus_{\omega \in \UU}L^2_{K,\omega}, \quad L^2_{K,\omega}: = \bigcap_{j=1}^\ell \ker_{L^2_K}((g_j)_* - \omega_j).
\end{equation}
It is worth noting that the spaces $L^2_{K,\omega}$ for $\omega \in \UU$ are pairwise orthogonal by virtue of the operators $(g_j)_*$ being unitary. 

Lastly, it will also simplify our later computations by introducing a convenient method of enumerating elements of $[K]$. Specifically, for each $j \in \JJ$ we define $m(j) \in \Z^n$ as the $n$-tuple satisfying
\begin{equation}\label{eq:mj}
    g^jK = K + m(j)\cdot (k_1, \ldots, k_n);
\end{equation}
then $m(j)$ exists and is unique by Assumption \ref{assump:3}.

\subsection{Strategy}\label{sec:3.4}
Our goal is to describe the structure of dispersion relations of $H_z$ near some quasi-momentum $K \in \mathbb{R}^n$ for generic values of $z$, where we continue to assume that $K$ together with a subgroup $G_0$ of $G$ satisfy Assumption \ref{assump:3}. The introduction of the parameter $z$ does not change the fact that, for $z \in \mathbb{R}$, $H_z$ is a self-adjoint unbounded operator on $L^2_k$ with compact resolvent (see section \ref{sec:3.1}). Since $\dom  H_z = H^2_k$ is independent of $z$ and $H_z \phi$ is linear in $z$ for any $\phi \in H^2_k$, it follows that $H_z$ is a self-adjoint holomorphic family of type (A), as per Definition \ref{def:typeA}, and thus we can apply Theorem \ref{thm:typeA} and Proposition \ref{prop:stronger}.

Building upon this, our strategy relies on the four key lemmas stated below, of which Theorem \ref{thm:1new} is an immediate consquence; their proofs are postponed to Section \ref{sec:4.5}. We will start with a result of eigenvalues of $-\Delta$ on $L^2_K$.

\begin{lemma}\label{lem:evectors}
Let $K \in \mathbb{R}^n$ and $G_0$ of $G$ satisfy Assumption \ref{assump:3}. For each $\omega \in \UU$, $\|K\|^2$ is an $L^2_{K,\omega}$-eigenvalue of $-\Delta$ of multiplicity 1, with corresponding normalized eigenvector given by
\begin{equation}\label{eq:evectors}
    \phi_\omega(x) = \frac{1}{\sqrt{|G_0|}}\sum_{j \in \JJ} \omega^j e^{i g^{-j}K\cdot x}.
\end{equation}
\end{lemma}

By Theorem \ref{thm:typeA}, there exists a function $\mu(z)$, analytic on $\mathbb{R}$, such that $\mu(z)$ is an $L^2_{K,\omega}$-eigenvalue of $H_z$ for $z \in \mathbb{R}$ and $\mu(0) = \|K\|^2$. Lemma \ref{lem:evectors} together with Corollary \ref{cor:simple} then enables us to compute the first order term in a Taylor expansion of $\mu(z)$.

\begin{lemma}\label{lem:perturbz}
Let $K \in \mathbb{R}^n$ and $G_0$ of $G$ satisfy Assumption \ref{assump:3} and let $\omega \in \UU$. There exist $\varepsilon,\delta > 0$ such that for $z \in (-\varepsilon, \varepsilon)$, $H_z$ has a unique $L^2_{K,\omega}$-eigenvalue in $(\|K\|^2\pm \delta)$, given by
\begin{equation}\label{eq:zperturbsum}
    \mu(z) = \|K\|^2 + z\cdot\sum_{j\in \JJ} \omega^j V_{m(j)}+ \cO(|z|^2),
\end{equation}
where $m(j)$ is the multi-integer defined in \eqref{eq:mj}.
\end{lemma}

By Theorem \ref{thm:typeA} and Proposition \ref{prop:stronger}, we can then conclude that, for generic $z \in \mathbb{R}$, $\mu(z)$ is a simple $L^2_{K,\omega}$-eigenvalue of $H_z$, splitting from the $L^2_K$-eigenvalue $\|K\|^2$ of $H_0 = -\Delta$, and the corresponding rank one eigenprojector can be extended to an analytic map on $\mathbb{R}$. 

When $K$ is a vertex of the Brillouin zone, we will be able to compute the generic multiplicities of the $L^2_K$-eigenvalues of $H_z$ splitting from the eigenvalue $\|K\|^2$ of $-\Delta$ using symmetry arguments. We will then describe the structure of the corresponding dispersion surfaces near $K$ using the following three results.

\begin{lemma}\label{lem:perturbk}
Let $\mu(z)$ be an $L^2_K$-eigenvalue of $H_z$ for some $z\in \mathbb{R}$, let $\pi(z): L^2_K \rightarrow L^2_K$ be the corresponding eigenprojector, and let $\cE(z)$ be the corresponding eigenspace. 
\begin{enumerate}[label=(\arabic*)]
    \item There exist $\varepsilon,\delta > 0$ such that for $\|\kappa\|< \varepsilon$, the $L^2_{K+\kappa}$ eigenvalues of $H_z$ in $\B_\delta(\mu(z))$ satisfy
    \begin{equation}\label{eq:det2}
        \det\big((u(z)-\mu) + M(z,\kappa) + R(\mu,\kappa)\big)|_{\cE(z)}=0,
    \end{equation}
    where $M(z,\kappa)=-\pi(z) (2i\kappa\cdot\nabla)\pi(z)$ and $\|R(\mu,\kappa)\| \leq C_1\|\kappa\|^2$ for some $C_1 > 0$.
    \item If $\lambda(z,\kappa)$ is a simple eigenvalue of $M(z,\kappa)$, continuous in $\kappa$ on some open set $U \subset B_\varepsilon(0)$ such that $\sup_{\kappa \in U}|\lambda(z,\kappa)| < \delta$, then there exists a simple eigenvalue $\mu(z,\kappa)$ of $H_z$ on $L^2_{K+\kappa}$ satisfying
    \begin{equation}\label{eq:relation}
        \mu(z,\kappa) = \mu(z) + \lambda(z,\kappa)+\cO(\|\kappa\|^2).
    \end{equation}
    \item If $M(z,\kappa) = 0$, then every $L^2_{K+\kappa}$-eigenvalue $\mu(z,\kappa)$ of $H_z$ satisfies $\mu(z,\kappa) = \mu(z)+ \cO(\|\kappa\|^2)$.
\end{enumerate}
\end{lemma}

For the following lemma, we continue to let $ M(z,\kappa):=-\pi(z) (2i\kappa\cdot\nabla)\pi(z)$, 
although we now allow $z$ to vary and let $\pi(z)$ denote the analytic family of orthogonal projections whose existence is guaranteed by Proposition \ref{prop:stronger} (which for generic $z \in \mathbb{R}$, is equal to the eigenprojector corresponding to $\mu(z)$).

\begin{lemma}\label{lem:mkanalytic} Let $\mu(z)$ be an $L^2_K$-eigenvalue of $H_z$, depending analytically on $z \in \mathbb{R}$. The characteristic polynomial of $M(z,\kappa)$, acting on the finite-dimensional space $\cE(z):=  \pi(z)(L^2_K)$, depends analytically on $z \in \mathbb{R}$.
\end{lemma}

Lastly, in order to compute the characteristic polynomial of $M(z,\kappa)$, we will express this matrix with respect to a basis consisting of one vector from each of the subspaces $L^2_{K,\omega}$ for $\omega$ in some subset of $\UU$. This final lemma will help us simplify these computations.

\begin{lemma}\label{lem:mkentries}
Let $\phi \in L^2_{K,\omega}$ and let $\psi \in L^2_{K,\widetilde{\omega}}$ for $\omega, \widetilde{\omega} \in \UU$. Then for all $j \in \JJ$, $\langle \phi,\nabla \psi \rangle$ is an eigenvector of $g^j$ with corresponding eigenvalue $\omega^{-j}\widetilde{\omega}^j$. Moreover, if $K$ is a vertex of the Brillouin zone $\cB$, then $\langle \phi,\nabla \phi\rangle = 0$.
\end{lemma}

\subsection{Proofs of Lemmas
\ref{lem:evectors} -- \ref{lem:mkentries} and Theorem \ref{thm:1new}}\label{sec:3.5}

\begin{proof}[Proof of Lemma \ref{lem:evectors}]
We first note that $\|K\|^2$, as an $L^2_K$-eigenvalue of $-\Delta$, by Assumption \ref{assump:3} has multiplicity $|G_0K| = |G_0|$. Consequently, it suffices to prove that for each $\omega \in \UU$, the function $\phi_\omega$ in \eqref{eq:evectors} is a normalized $L^2_{K,\omega}$-eigenvector for the eigenvalue $\|K\|^2$, for then this eigenvalue on $L^2_{K,\omega}$ would necessarily be simple since $|\UU| = |G_0|$.

To see that $\phi_\omega$ is normalized, observe that the $|G_0|$ functions $e^{ig^jK\cdot x}$ form an orthonormal system because of Assumption \ref{assump:3}. Indeed, each of these exponentials is distinct, for if $g^jK = g^{j'}K$, then $g^{j-j'}K =K$, which implies $g^{j-j'} = I$ because $|G_0K|=|G_0|$. Therefore $g^j = g^{j'}$ and $j = j'$.

To show that $\phi_\omega \in L^2_{K,\omega}$, we first note that $\phi_\omega \in L^2_K$ since $g^jK \in K + \Lambda^*$ for every $j \in \JJ$. In addition, we compute that
\begin{equation}
    (g_1)_*\phi_\omega(x)  = \frac{1}{\sqrt{|G_0|}}\sum_{j \in \JJ} \omega^j  e^{ig^{-j}K\cdot g^\top x} = \frac{\omega_1}{\sqrt{|G_0|}}\sum_{j \in \JJ} \omega^{j-e_1}  e^{ig^{-j+e_1}K\cdot x} = \omega_1\phi_\omega(x),
\end{equation}
with similar identities when testing the pushforward operators by $g_2,\ldots,g_\ell$. We conclude by noting that $\phi_\omega$ is an eigenvector of $-\Delta$ with eigenvalue $\|K\|^2$ since $\|g^jK\| = \|K\|$ by virtue of $g^j$ being an orthogonal matrix for every $j \in \JJ$.
\end{proof}

\begin{proof}[Proof of Lemma \ref{lem:perturbz}]
By Lemma \ref{lem:evectors} and Corollary \ref{cor:simple}, there exist $\varepsilon,\delta > 0$ such that $H_z$ has a single eigenvalue $\mu(z)$ in $(\|K\|^2\pm \delta)$ for all $z \in (-\varepsilon,\varepsilon)$, given by:
\begin{align}\label{eq:zperturbprod}
    \mu(z) & =  \|K\|^2 + z\langle \phi_\omega, H_z'\phi_\omega\rangle + \cO(|z|^2)\\
    & =  \|K\|^2 + z\langle \phi_\omega, V\phi_\omega\rangle + \cO(|z|^2).
\end{align}
It therefore suffices to prove that 
\begin{equation}
     \langle \phi_\omega, V\phi_\omega \rangle = \sum_{j\in\JJ} \omega^j V_{m(j)}.
\end{equation}

Towards that end, first recall that $V$, viewed as a multiplication operator, commutes with $g_*$ for all $g \in G$ by virtue of $V$ being $\Lambda$-invariant. Consequently, by expanding $\phi_\omega$ via \eqref{eq:evectors}, we compute that
\begin{align}
    \langle \phi_\omega, V\phi_\omega \rangle & = \frac{1}{\sqrt{|G_0|}}\sum_{j \in \JJ} \omega^{-j}\left\langle e^{ig^{-j}K\cdot x},V\phi_\omega\right\rangle \\
    & = \frac{1}{\sqrt{|G_0|}}\sum_{j \in \JJ} \omega^{-j}\left\langle e^{iK\cdot x},V(g_*^j\phi_\omega)\right\rangle 
    = \frac{1}{\sqrt{|G_0|}}\sum_{j \in \JJ} \left\langle e^{iK\cdot x},V\phi_\omega\right\rangle 
    = \sqrt{|G_0|}\left\langle e^{iK\cdot x},V\phi_\omega\right\rangle\\ %
    & = \sum_{j \in \JJ}\omega^j\left\langle e^{iK\cdot x},Ve^{ig^{-j}K\cdot x}\right\rangle
    = \sum_{j \in \JJ}\omega^j\left\langle e^{i(K-g^{-j}K)\cdot x},V\right\rangle
    =\sum_{j\in\JJ} \omega^j V_{-m(j)}. \label{eq:fourierindex}
\end{align}
Note that in \eqref{eq:fourierindex}, we have used the fact that $K-g^{-j}K = - m(j)\cdot (k_1,\ldots,k_n)$. Lastly, since $V$ is necessarily even, $V_{-m(j)} = V_{m(j)}$ for all $j \in \JJ$, thus completing the proof.
\end{proof}

\begin{proof}[Proof of Lemma \ref{lem:perturbk}]
We first prove, using the Schur complement formula, that there exist $\varepsilon,\delta > 0$ such that for $\|\kappa\|< \varepsilon$ and $\mu \in \B_\delta(\mu(z))$, 
\begin{align}
    & H_z - \mu \text{ is invertible on }L^2_{K+\kappa}\\
    \Leftrightarrow\;  & (u(z)-\mu) + M(z,\kappa) + R(\mu,\kappa) \text{ is invertible on }\cE(z), \label{eq:complement}
\end{align}
where $M(z,\kappa)=-\pi(z) (2i\kappa\cdot\nabla)\pi(z)$ and $R(\mu,\kappa) \leq C_1\|\kappa\|^2$ for some $C_1 > 0$. Since $z \in \mathbb{R}$ is assumed to be fixed, for simplicity of notation we suppress the dependence of $\cE(z)$ on $z$, and denote this eigenspace simply by $\cE$. We also note that the operators $H_z$ on $L^2_{K+\kappa}$ and $H_{z,\kappa} := e^{-i\kappa\cdot x}H_ze^{i\kappa\cdot x}$ on $L^2_K$ have the same spectrum. Indeed, if $\phi(x) \in L^2_K$, then $\psi(x): = e^{i\kappa\cdot x}\phi(x) \in L^2_{K+\kappa}$. Furthermore, $\psi$ is an $L^2_{K+\kappa}$ eigenvector of $H_z$ with eigenvalue $\mu$ if and only if 
\begin{equation}
    H_{z,\kappa}\phi(x) = e^{-i\kappa\cdot x}H_z\psi(x) = \mu\phi(x),
\end{equation}
or, in other words, $\phi$ is an $L^2_K$ eigenvector of $H_{z,\kappa}$ with eigenvalue $\mu$. Therefore, it is equivalent to prove that $H_{z,\kappa} - \mu$ is invertible on $L^2_K$ if and only if $(u(z)-\mu) + M(z,\kappa) + R(\mu,\kappa) $ is invertible on $\cE$.

Write $H_{z,\kappa}$ as a $2\times 2$ block operator with respect to the decomposition $L^2_K = \cE\oplus \cE^\perp$:
\begin{equation}
    H_{z,\kappa} = \begin{pmatrix}[1.5]
        H_{z,\kappa}^{(11)} & H_{z,\kappa}^{(12)} \\
        H_{z,\kappa}^{(21)} & H_{z,\kappa}^{(22)}
    \end{pmatrix}.
\end{equation}
Letting $\pi(z)^\perp = I - \pi(z)$, we compute that
\begin{align}
    H_{z,\kappa}^{(11)} & = \pi(z)H_{z,\kappa}\pi(z) = \pi(z) e^{-i\kappa\cdot x}(-\Delta +zV(x))e^{i\kappa\cdot x}\pi(z)  \\
    & = \pi(z)(H_z - 2i\kappa \cdot \nabla+\|\kappa\|^2)\pi(z) \\
    & = \mu(z) + M(z,\kappa)+\|\kappa\|^2,\\
    H_{z,\kappa}^{(12)} & = \pi(z)(H_z - 2i\kappa \cdot \nabla+\|\kappa\|^2)\pi(z)^\perp \\
    & = -\pi(z)(2i\kappa \cdot \nabla)\pi(z)^\perp = \cO(\|\kappa\|),\\
    H_{z,\kappa}^{(21)} & = \left(H_{z,\kappa}^{(12)}\right)^* = \cO(\|\kappa\|).
\end{align}

Next, we claim that there exist $\varepsilon, \delta > 0$ such that for $\|\kappa\|< \varepsilon$ and $\mu \in \B_\delta(\mu(z))$, $H^{(22)}_{z,\kappa} - \mu$ is invertible and its inverse is uniformly bounded in $\mu$. To prove this, observe that since $H_z$ has a compact resolvent and is self-adjoint due to our assumption that $z \in \mathbb{R}$, we can order the distinct eigenvalues of $H_z$ so that there exist eigenvalues $\mu_- < \mu(z) < \mu_+$ and the remaining eigenvalues of $H_z$ are all strictly farther away from $\mu(z)$. Consequently, if we let
\begin{equation}
    \delta = \frac{1}{2}\min\big\{|\mu(z)-\mu_-|,|\mu(z)-\mu_+|\big\},
\end{equation}
it then follows that $H_z|_{\cE^\perp}$ has no eigenvalues in  $\B_\delta(\mu(z))$ since by construction $\mu(z)$ is not an eigenvalue of $H_z|_{\cE^\perp}$. Therefore $H_z|_{\cE^\perp} - \mu$ is invertible and its inverse satisfies $\|(H_z-\mu)|_{\cE^\perp}^{-1}\| \leq 1/\delta$. In addition, the fact that $H_z$ is self-adjoint implies that $H_z|_{\cE^\perp}(\cE^\perp)\subset \cE^\perp$, and as a result $(H_z - \mu)|_{\cE^\perp}^{-1}$ is a well-defined operator from $\cE^\perp$ to itself (and in fact is a bijection from $\cE^\perp$ to $H^2_K \cap \cE^\perp$ by elliptic regularity). 

Thus, for all $\mu \in \B_\delta(\mu(z))$, we have
\begin{align}
    H^{(22)}_{z,\kappa} - \mu & = \pi(z)^\perp(H_z -\mu - 2i\kappa \cdot \nabla+\|\kappa\|^2)\pi(z)^\perp  \\
    & = (H_z - \mu)|_{\cE^\perp} - \pi(z)^\perp (2i\kappa \cdot \nabla)\pi(z)^\perp+\|\kappa\|^2 \\
    & = (H_z - \mu)|_{\cE^\perp}\left(I - T(\mu,\kappa)\right),
\end{align}
where
\begin{equation}
    T(\mu,\kappa)  = (H_z - \mu)|_{\cE^\perp}^{-1}\left(\pi(z)^\perp (2i\kappa \cdot \nabla)\pi(z)^\perp+\|\kappa\|^2\right).
\end{equation}
Again using elliptic regularity, for each $\mu \in \B_\delta(\mu(z))$ there exists some $C_\mu>0$ such that $\|T(\mu,\kappa)\|\leq C_\mu\|\kappa\|$. In addition, since $\B_\delta(\mu(z))$ is contained in the resolvent set of $H_z$, $T(z,\kappa)$ is continuous in $\mu$ on this set, which together with the fact that $\B_\delta(\mu(z))$ is precompact, means that there exists some $C >0$, uniform in $\mu$, such that $\|T(\mu,\kappa)\|\leq C\|\kappa\|$. Therefore, if we set $\varepsilon = 1/C$, it follows from a Neumann series argument such as the one following \eqref{eq:neuman} that $H^{(22)}_{z,\kappa} - \mu$ is invertible and its inverse satisfies
\begin{equation}
    (H^{(22)}_{z,\kappa} - \mu)^{-1} = (H_z - \mu)|_{\cE^\perp}^{-1}+\cO(\|\kappa\|),
\end{equation}
uniformly in $\mu$, for $\mu \in \B_\delta(\mu(z))$ and $\|\kappa\| < \varepsilon$. 

Since $H^{(22)}_{z,\kappa} - \mu$ is invertible for all $\mu \in \cB_\delta(\mu(z)$ and $\|\kappa\| < \varepsilon$, the Schur complement of the block $H^{(22)}_{z,\kappa} - \mu$ is well-defined for all such $\mu$ and $\kappa$ and is given by:
\begin{align}
    (H^{(11)}_{z,\kappa} - \mu) + H^{(12)}_{z,\kappa}(H^{(22)}_{z,\kappa} - \mu)^{-1}H^{(21)}_{z,\kappa} & = (\mu(z)-\mu) + M(z,\kappa) + R(\mu,\kappa), \quad \text{where}\\
    R(\mu,\kappa) & = \|\kappa\|^2 + H^{(12)}_{z,\kappa}(H^{(22)}_{z,\kappa} - \mu)^{-1}H^{(21)}_{z,\kappa}. \label{eq:schurcomplement}
\end{align}
However, since $\|(H^{(22)}_{z,\kappa} - \mu)^{-1}\| \leq 1/\delta + C\|\kappa\|$ uniformly in $\mu$, it follows that $\|R(\mu,\kappa )\| \leq C_1\|\kappa\|^2$ for some $C_1 > 0$, thus proving \eqref{eq:complement}. Using this, we now prove (1), (2), and (3) of Lemma \ref{lem:perturbk}.

(1) This is immediate from \eqref{eq:complement}, since $\mu$ is an $L^2_{K+\kappa}$-eigenvalue of $H_z$ if and only if $H_z - \mu$ is not invertible, and $(\mu(z)-\mu) + M(z,\kappa) + R(\mu,\kappa)$ is not invertible on $\cE$ if and only if its determinant is zero.

(2) Suppose $\lambda(z,\kappa)$ is a simple eigenvalue of $M(z,\kappa)$, continuous in $\kappa$ on some open set $U \subset B_\varepsilon(0)$ such that $\sup_{\kappa \in U}|\lambda(z,\kappa)| < \delta$. Then for any $\kappa_0 \in U$, by continuity of $\lambda$ there exists a neighborhood $U_0 \subset U$ of $\kappa_0$ and a $ \delta_0 < \delta$ such that $\sup_{\kappa \in U_0}|\lambda(z,\kappa)| < \delta_0$. Thus, by simplicity, there exists a simple, closed, positively-oriented contour $\cC$ contained in $\B_\delta(\mu(z))$, such that $\cC$ strictly encloses $\mu(z) + \lambda(z,\kappa)$ and no other eigenvalue of $\mu(z) + M(z,\kappa)$ for all $\kappa \in U_0$. In addition, since $\cC$ and $\overline{U_0}$ are compact (since $U_0 \subset B_\varepsilon(0)$) and $((\mu(z) - \mu) + M(z,\kappa))^{-1}$ is continuous in $\mu,\kappa$ for all $\mu\in \cC$ and $\kappa \in U_0$, there exists $C_2 > 0$ such that 
\begin{equation}\label{eq:Mbound}
    \|((\mu(z) - \mu) + M(z,\kappa))^{-1}\| \leq C_2
\end{equation}
for all $\mu \in \cC$ and $\kappa \in U_0$.

We now want to use Cauchy's integral formula to relate the eigenvalues of $H_{z,\kappa}$ to those of $M(z,\kappa)$. To do so, we now prove that the operator $(\mu(z)-\mu) + M(z,\kappa) + R(\mu,\kappa)$ is invertible and its inverse is uniformly bounded in $\mu$ for all $\mu \in \cC$ and $\kappa \in U_0$. First, observe that
\begin{align}
    (\mu(z)-\mu) + M(z,\kappa) + R(\mu,\kappa) & = ((\mu(z)-\mu) + M(z,\kappa))\left(I+ S(z,\kappa)\right), \quad \text{ where } \\
    S(z,\kappa) & = ((\mu(z)-\mu) + M(z,\kappa))^{-1}\cdot \cO(\|\kappa\|^2).
\end{align}
By \eqref{eq:Mbound}, after increasing $C_2$ if necessary, $\|S(z,\kappa)\| \leq C_2\|\kappa\|^2$ for all $\mu \in \cC$. Therefore, by replacing $\varepsilon$ with the minimum of itself and $1/\sqrt{C_2}$, it follows from another Neumann series argument that $(\mu(z)-\mu) + M(z,\kappa) + R(\mu,\kappa)$ is invertible and its inverse satisfies
\begin{equation}
    \left(((\mu(z)-\mu) + M(z,\kappa) + R(\mu,\kappa))\right)^{-1} = ((\mu(z)-\mu) + M(z,\kappa))^{-1} + \cO(\|\kappa\|^2)
\end{equation}
for all $\mu \in \cC$ and $\kappa\in U_0$, where again the bound is uniform in $\mu$.

Thus, for all such $\mu$ and $\kappa$, we can write $(H_{z,\kappa} -\mu)^{-1}$ with respect to the decomposition $L^2_K = \cE\oplus \cE^\perp$ as
\begin{equation}
    (H_{z,\kappa} -\mu)^{-1} = \begin{pmatrix}[1.5]
         ((\mu(z)-\mu) + M(z,\kappa))^{-1} + \cO(\|\kappa\|^2) & \cO(\|\kappa\|) \\
        \cO(\|\kappa\|) & (H^{(22)}_{z,\kappa} - \mu)^{-1}+\cO(\|\kappa\|^2)
    \end{pmatrix},
\end{equation}
where all bounds are uniform in $\mu$. Consequently, by applying Cauchy's integral formula and taking the trace of both sides, we get
\begin{align}
    \Tr \bigg(\frac{1}{2\pi i}\oint_\cC \mu(H_{z,\kappa} -&  \mu)^{-1}d\mu \bigg) \\
    & = \Tr \left(\frac{1}{2\pi i}\oint_\cC \mu((\mu(z)-\mu) + M(z,\kappa))^{-1}d\mu \right )\\
    & \quad + \Tr \left(\frac{1}{2\pi i}\oint_\cC \mu(H^{(22)}_{z,\kappa} - \mu)^{-1}d\mu \right) + \Tr \left(\frac{1}{2\pi i}\oint_\cC \mu\cdot \cO(\|\kappa\|^2)d\mu \right)\\
    & = \mu(z) + \lambda(z,\kappa)+\cO(\|\kappa\|^2). \label{eq:laststep}
\end{align}
To compute  \eqref{eq:laststep}, we have used three facts. First, we used that $\lambda(z,\kappa)$ is a simple eigenvalue of $M(z,\kappa)$ and the only eigenvalue contained in $\cC$. Second, we used that $\mu(H^{(22)}_{z,\kappa} - \mu)^{-1}$ is analytic in $\mu$ on $\B_\delta(\mu(z))$, and so its integral on $\cC$ equals zero. Lastly, we used that the integral of $\mu\cdot\cO(\|\kappa\|^2)$ is $\cO(\|\kappa\|^2)$, due to the ML inequality and since the bound is uniform in $\mu$. An identical argument also tells us that 
\begin{equation}
    \Tr \left(\frac{1}{2\pi i}\oint_\cC (H_{z,\kappa} -  \mu)^{-1}d\mu \right )  = \Tr \left(\frac{1}{2\pi i}\oint_\cC ((\mu(z)-\mu) + M(z,\kappa))^{-1}d\mu \right) +\cO(\|\kappa\|^2). \label{eq:laststep2}
\end{equation}

From here, note that $H_{z,\kappa}$ is self-adjoint since it is unitarily equivalent to $H_z$, and it is an analytic family of type (A) in each component of $\kappa = (\kappa_1,\ldots,\kappa_n)$, as per Definition \ref{def:typeA}. Consequently, by the residue theorem, it follows that if we let 
\begin{equation}
    \pi(\kappa) = \frac{1}{2\pi i}\oint_\cC (H_{z,\kappa} -  \mu)^{-1}d\mu ,
\end{equation}
then $\pi(\kappa)$ is the projection onto the eigenspaces corresponding to eigenvalues of $H_{z,\kappa}$ contained in $\cC$, and it is analytic in each component of $\kappa$ since its integrand is. Therefore \eqref{eq:laststep2} becomes
\begin{equation}
    \rank \pi(\kappa) = 1 + \cO(\|\kappa\|^2).
\end{equation}
Since $\pi(\kappa)$ is analytic in each component of $\kappa$, its rank must be constant, and we therefore deduce that $H_{z,\kappa}$ has a single, simple eigenvalue $\mu(z,\kappa)$ in $\cC$. It then follows from \eqref{eq:laststep} that 
\begin{equation}
    \mu(z,\kappa) = \mu(z) + \lambda(z,\kappa)) +\cO(\|\kappa\|^2)
\end{equation}
for all $\kappa \in U_0$. Since this equation holds on a neighborhood of $\kappa_0$ for every $\kappa_0 \in U$, we conclude that \eqref{eq:relation} holds on all of $U$.

(3) Assume that $M(z,\kappa) = 0$; then by (1), for $\|\kappa \| < \varepsilon$, the $L^2_{K+\kappa}$-eigenvalues of $H_z$ in $B_\delta(\mu(z))$ are equal to the eigenvalues of $\mu(z) + R(\mu,\kappa)|_\cE$, which in turn are equal to $\mu(z)$ plus the eigenvalues of $R(\mu,\kappa)$. However, if $\lambda(\kappa)$ is an eigenvalue of $R(\mu,\kappa)|_\cE$, then 
\begin{equation}
    |\lambda(\kappa)| \leq \|R(\mu,\kappa)\| \leq C_1\|\kappa\|^2.
\end{equation}
Therefore, for $\|\kappa \| < \varepsilon$, the $L^2_{K+\kappa}$-eigenvalues of $H_z$ in $B_\delta(\mu(z))$ satisfy $\mu(z,\kappa) = \mu(z) + \cO(\|\kappa\|^2)$. 
\end{proof}

\begin{proof}[Proof of Lemma \ref{lem:mkanalytic}]
Let $m = \dim \cE(z)$; by Proposition \ref{prop:stronger}, $m$ is independent of $z$. As a symmetric function of eigenvalues, the determinant of $M(z,\kappa) -\lambda$ can be expressed as a (universal) polynomial in the traces of
its $m$ first powers. This means that $\det_{\cE(z)}(M(z,\kappa) -\lambda)$ is polynomial (with coefficients independent of $z$, $\kappa$, and $\lambda$) in
\begin{equation}
    \Tr_{\cE(z)}\big( (M(z,\kappa) - \lambda)^j \big) = \Tr_{L^2_K}\big( (M(z,\kappa) - \lambda \pi(z) )^j \big), \quad \text{ for }j=1,\ldots, m.
\end{equation}
The operator $M(z,\kappa) - \lambda \pi(z)$ is finite-rank and analytic in $z$, and hence its trace is analytic in $z$. Thus $\det_{\cE(z)} \big( M(z,\kappa) - \lambda \big)$ is analytic in $z$. 
\end{proof}

\begin{proof}[Proof of Theorem \ref{thm:1new}]
For the sake of thoroughness, we remark that Theorem \ref{thm:1new} is then an immediate consequence of Lemma \ref{lem:perturbk}(1) and Lemma \ref{lem:mkanalytic}. 
\end{proof}

\begin{proof}[Proof of Lemma \ref{lem:mkentries}]
We start by looking at how the group action of $G_0$ interacts with the gradient of a function $f\in L^2_{K}$. Let $g \in G_0$; then
\begin{equation}
    \nabla( g_*f)(x) = \nabla (f(g^\top x)) = g(\nabla f)(g^\top x) = g(g_*\nabla f)(x). 
\end{equation}
Multiplying the first and last of these expressions on the right by $g^\top$, we get 
\begin{equation}\label{eq:gradaction}
    g_*\nabla f = g^\top\nabla(g_*f).
\end{equation}

Now let $\phi \in L^2_{K,\omega}$ and let $\psi \in L^2_{K,\widetilde{\omega}}$ for $\omega, \widetilde{\omega} \in \UU$. To show that $\langle \phi,\nabla \psi \rangle$ is an eigenvector of $g^j$ with corresponding eigenvalue $\omega^{-j}\widetilde{\omega}^j$ for all $j \in \JJ$, we compute the following, using \eqref{eq:gradaction} and the fact that $(g^j)^\top = g^{-j}$ since $G_0$ consists of orthogonal matrices:
\begin{equation}\label{eq:offdiagonal}
    g^j\langle \phi,\nabla \psi \rangle = g^j\langle g^j_*\phi,g^j_*\nabla \psi \rangle = \langle g^j_*\phi,\nabla(g^j_* \psi) \rangle =  \overline{\omega^{j}}\widetilde{\omega}^j\langle \phi,\nabla \psi \rangle = \omega^{-j}\widetilde{\omega}^j\langle \phi,\nabla \psi \rangle.  
\end{equation}

Lastly, to show that $\langle \phi,\nabla \phi \rangle = 0$ when $K$ is a vertex of $\cB$, note that \eqref{eq:offdiagonal} implies that $\langle \phi,\nabla \phi \rangle$ is an eigenvector of $g^j$ with eigenvalue 1 for all $j \in \JJ$.  However, we claim that the only such vector is the zero vector.  Suppose $k \in \mathbb{R}^n$ such that $gk$ for all $g \in G_0$. Then, since $K$ is assumed to be a vertex of $\cB$, it necessarily must lie on at least $n$ hyperfaces of $\cB$, and therefore, by Proposition \ref{prop:brillouinmult}, there exist linearly independent lattice vectors $K_1,\ldots,K_n \in \Lambda^*$ such that $K - K_j$ is also a vertex of $\cB$ satisfying $\|K-K_j\| = \|K\|$, and therefore contained in the equivalence class $[K]$. By Assumption \ref{assump:3}, for $j=1,\ldots, n$, there exists $h_j \in G_0$ (where the notation here is chosen so as to differentiate $h_j$ from the generator $g_j$) such that $h_jK = K-K_j$. Thus we get that, for all $j$, 
\begin{equation}
    k\cdot K_j = k\cdot (K- h_jK) = k\cdot K - h_j^\top k \cdot K = k\cdot K - k \cdot K = 0. 
\end{equation}
Since the set $\{K_j\}_{j=1}^n$ is linearly independent, it is a basis for $\mathbb{R}^n$, and it therefore follows that $k = 0$. As a result, we conclude that $\langle \phi,\nabla \phi\rangle  = 0$.
\end{proof}

\section{Schr{\"o}dinger Operators Invariant Under Cubic Lattices}\label{chp:4}

In this section, we focus on Schr{\"o}dinger operators invariant under {\it cubic lattices}, which are lattices whose point groups are isomorphic to the {\it octahedral group}. Every such lattice is isometric, up to a dilation, to one of the three lattices generated by the bases listed in row 1 of Table \ref{tab:geometry}; these lattices are called the {\it simple cubic}, {\it body-centered cubic}, and {\it face-centered cubic}, respectively. Using the general theory developed in \S\ref{chp:3}, we prove that the dispersion surfaces of such Schr{\"o}dinger operators generically have unusual dispersion surfaces near vertices of the Brillouin zone: Theorem \ref{thm:2}.

\subsection{Geometry of Cubic Lattices.}\label{sec:4.1}

\begin{figure}[!ht]
    \centering
    \begin{tabular}{|c|c|c|c|c|}
        \hline
        &  & \multicolumn{2}{|c|}{} & \\[-1em]
        & $\Lambda^{S}$ & \multicolumn{2}{|c|}{$\Lambda^{BC}$} & $\Lambda^{FC}$ \\
        \hline
        & & \multicolumn{2}{|c|}{} & \\
        \text{Basis} & $\begin{pmatrix}1 \\ 0 \\ 0\end{pmatrix}, \begin{pmatrix}0 \\ 1 \\ 0\end{pmatrix},\begin{pmatrix}0 \\ 0 \\ 1\end{pmatrix}$ & 
        \multicolumn{2}{|c|}{$\begin{pmatrix}1 \\ 0 \\ 0\end{pmatrix}, \begin{pmatrix}0 \\ 1 \\ 0\end{pmatrix},\begin{pmatrix}1/2 \\ 1/2 \\ 1/2\end{pmatrix}$} & 
        $\begin{pmatrix}1/2 \\ 1/2 \\ 0\end{pmatrix}, \begin{pmatrix}-1/2 \\ 1/2 \\ 0\end{pmatrix},\begin{pmatrix}0 \\ -1/2 \\ 1/2\end{pmatrix}$ \\
        & & \multicolumn{2}{|c|}{} & \\
        \hline
        & & \multicolumn{2}{|c|}{} & \\
            \text{Dual Basis} & $\begin{pmatrix}2\pi \\ 0 \\ 0\end{pmatrix}, \begin{pmatrix}0 \\ 2\pi \\ 0\end{pmatrix},\begin{pmatrix}0 \\ 0 \\ 2\pi\end{pmatrix}$ & 
            \multicolumn{2}{|c|}{$\begin{pmatrix}2\pi \\ 0 \\ -2\pi\end{pmatrix}, \begin{pmatrix}0 \\ 2\pi \\ -2\pi\end{pmatrix},\begin{pmatrix}0 \\ 0 \\ 4\pi\end{pmatrix}$} & 
        $\begin{pmatrix}2\pi\\ 2\pi \\ 2\pi\end{pmatrix}, \begin{pmatrix}-2\pi \\ 2\pi \\ 2\pi\end{pmatrix},\begin{pmatrix}0 \\ 0 \\ 4\pi\end{pmatrix}$ \\
        & & \multicolumn{2}{|c|}{} & \\
        \hline
        & & \multicolumn{2}{|c|}{} & \\
        $\cB$ & \includegraphics[width=0.17\textwidth]{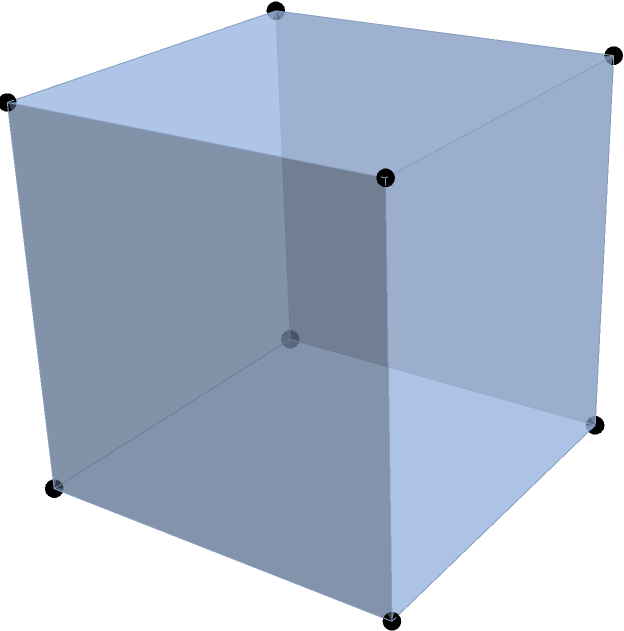} &\multicolumn{2}{|c|}{\includegraphics[width=0.17\textwidth]{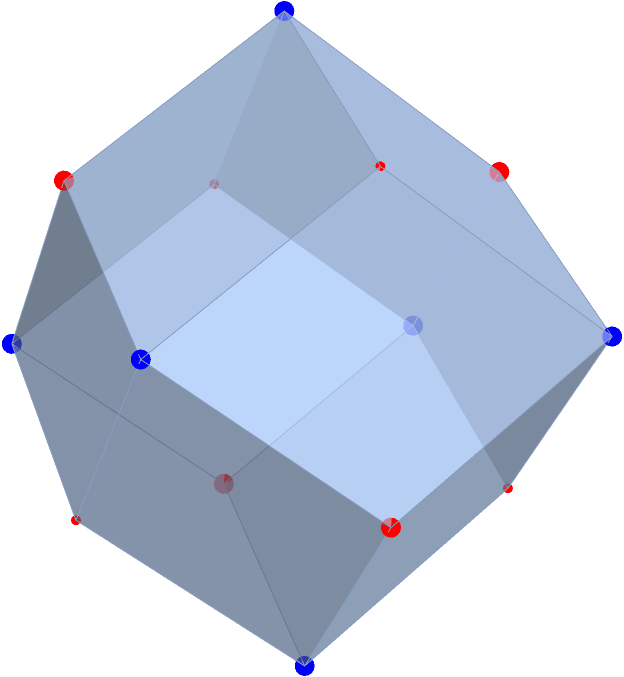}} & \includegraphics[width=0.17\textwidth]{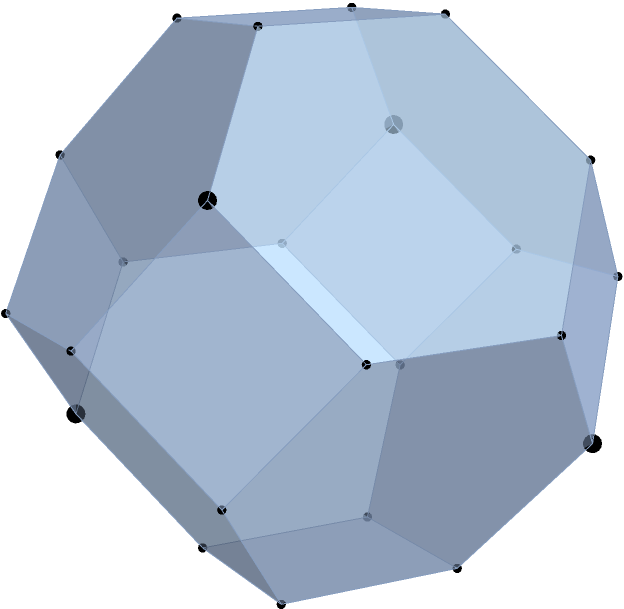} \\
        & & \multicolumn{2}{|c|}{} & \\
        \hline
        & & & & \\
        $K$ & $\begin{pmatrix} \pi \\ \pi \\ \pi \end{pmatrix}$ & $\begin{pmatrix}  \pi \\ \pi \\ \pi  \end{pmatrix}$ &  $\begin{pmatrix} 0 \\0 \\ 2\pi \end{pmatrix}$ &  $\begin{pmatrix} 0 \\ 2\pi \\ \pi  \end{pmatrix}$  \\
        & & & & \\
        \hline
        & & & &  \\
       $m$ & 8 & 4 & 6 & 4 \\
        & & & & \\
       \hline
        & & & &  \\
        $G_0$ & $\langle f_1,f_2,f_3\rangle$ 
        & $\langle f_{13},f_{23}\rangle$ 
        & $\langle r,f\rangle$ 
        & $\langle s_0 \rangle$ 
        \\
        & & & &  \\
        \hline
        & & & &  \\
        $\UU$ & $U_2^3$ & $U_2^2$ & $ U_3 \times U_2$ &  $U_4$ \\
        & & & &  \\
        \hline
    \end{tabular}
    \caption{Geometry of the cubic lattices.}
    \label{tab:geometry}
\end{figure}

Let $\Lambda^{S}$, $\Lambda^{BC}$, and $\Lambda^{FC}$ denote the simple cubic, body-centered cubic, and face-centered cubic lattice, respectively, which are generated by the bases given in row 1 of Table \ref{tab:geometry}. We then give spectral results for $-\Delta$ seen as a $\Lambda^S$, $\Lambda^{BC}$, and $\Lambda^{FC}$ invariant operator on $L^2_K$, where $K$ is a vertex of the corresponding Brillouin zone $\cB$, as these points exhibit a high degree of symmetry (see Proposition \ref{prop:brillouinmult}). In addition, as noted in \S \ref{sec:3.2}, it suffices to consider vertices $K$ which have distinct orbits under the action of the point group $G$.

As previously discussed, the point group of the lattices $\Lambda^S, \Lambda^{BC}$ and $\Lambda^{FC}$ is the octahedral group, which we denote by $G$, and which is generated by the three matrices
\begin{equation}\label{eq:generators}
    r := \begin{pmatrix}0 & 0 & 1 \\ 1 & 0 & 0 \\ 0 & 1 & 0  \end{pmatrix}, \quad
    s := \begin{pmatrix}0 & 1 & 0 \\ 1 & 0 & 0 \\ 0 & 0 & -1  \end{pmatrix}, \quad \text{and}\quad 
    f := \begin{pmatrix}-1 & 0 & 0 \\ 0 & -1 & 0 \\ 0 & 0 & -1  \end{pmatrix}.
\end{equation}
We shall also later need the following elements of $G$:
\begin{equation}\label{eq:elements1}
    f_1:= \begin{pmatrix}-1 & 0 & 0 \\ 0 & 1 & 0 \\ 0 & 0 & 1  \end{pmatrix}, \quad f_2:=\begin{pmatrix}1 & 0 & 0 \\ 0 & -1 & 0 \\ 0 & 0 & 1  \end{pmatrix}, \quad f_3:=\begin{pmatrix}1 & 0 & 0 \\ 0 & 1 & 0 \\ 0 & 0 & -1  \end{pmatrix}, 
\end{equation}
\begin{equation}\label{eq:elements2}
    f_{12}:= \begin{pmatrix}-1 & 0 & 0 \\ 0 & -1 & 0 \\ 0 & 0 & 1  \end{pmatrix}, \quad f_{13}:=\begin{pmatrix}-1 & 0 & 0 \\ 0 & 1 & 0 \\ 0 & 0 & -1  \end{pmatrix}, \quad f_{23}:=\begin{pmatrix}1 & 0 & 0 \\ 0 & -1 & 0 \\ 0 & 0 & -1  \end{pmatrix}, 
\end{equation}
\begin{equation}\label{eq:elements3}
     \text{ and } \quad s_0 := \begin{pmatrix}0 & 1 & 0 \\ -1 & 0 & 0 \\ 0 & 0 & -1  \end{pmatrix}.
\end{equation}

In Table \ref{tab:geometry}, in addition to a basis $v_1, v_2, v_3$, we list for each of the lattices $\Lambda^S$, $\Lambda^{BC}$, and $\Lambda^{FC}$:
\begin{itemize}
    \item The corresponding dual basis $k_1, k_2, k_3$ (which by definition satisfies $v_j \cdot k_\ell = 2\pi \delta_{j\ell}$ for $j,\ell=1,2,3$);
    \item A picture of the Brillouin zone $\cB$, where the vertices are colored and sized in reference to the vertices listed in the following row. Specifically, given a vertex $K$, the set $[K]$ (defined in \eqref{eq:equivclass}) consists of vertices of $\cB$ by Proposition \ref{cor:brillouinvert}, which are colored the same and have larger dots. Those vertices which lie in the same orbit under $G$ but are not in $[K]$ are colored the same but have smaller dots;
    \item Vertices $K$ of the Brillouin zone, corresponding to distinct orbits under the action of $G$.
    \item The multiplicity $m$ of the $L^2_K$-eigenvalue $\| K\|^2$ of $-\Delta$, equal to the cardinality of the set $[K]$;
    \item An abelian subgroup $G_0$ of $G$, expressed in terms of its generators, which together with the vertex $K$ satisfy Assumption \ref{assump:3};
    \item The corresponding group $\UU$ consisting of tuples of roots of unity, as defined in \eqref{eq:JU}.
\end{itemize}

\subsection{Proof Outline for Theorem \ref{thm:2}}\label{sec:4.2}

In each of the following four sections, we shall prove Theorem \ref{thm:2} for one of the three cubic lattices together with one of the vertices $K$ listed in row 4 of Table \ref{tab:geometry} by using the lemmas stated in \S \ref{sec:3.4}. Each of these proofs will require the same three steps, which we now outline.

{\it (1) Upper bound on multiplicity:} Let $\Lambda$ be one of the three cubic lattices listed in Table \ref{tab:geometry}, let $K$ be one of the listed vertices for this lattice, and let $(k_1,k_2,k_3)$, $\cB$, $m$, $G_0$, and $\UU$ be the objects listed in the column corresponding to this vertex. We also let $V$ be a $\Lambda$-invariant potential and let $H_z = -\Delta+z V$.  Then for each $\omega \in \UU$, Lemma \ref{lem:perturbz} describes how the multiplicity $m$, $L^2_K$-eigenvalue $\|K\|^2$ of $H_0 =-\Delta$ splits as $z$ increases into simple $L^2_{K,\omega}$-eigenvalues given by:
\begin{equation}
    \mu_\omega(z) = \|K\|^2 + z \cdot \sum_{j \in \JJ}\omega^jV_{m(j)} + \cO(|z|^2),
\end{equation}
In particular, if $\omega, \widetilde{\omega} \in \UU$ are such that $\mu'_\omega(0), \mu'_{\widetilde{\omega}}(0)$ are distinct, then the eigenvalues $\mu_\omega(z), \mu_{\widetilde{\omega}}(z)$ clearly split. This test provides an upper bound on possible multiplicities of $\mu_\omega(z)$, viewed as an $L^2_K$-eigenvalue.

{\it (2) Lower bound on multiplicity:} Our argument in step (1) is inconclusive when $\mu'_\omega(0) = \mu'_{\widetilde{\omega}}(0)$, and so in this case we provide a lower bound on the splitting multiplicities using a symmetry argument. Note that the multiplicity of $\mu_\omega(z)$ as an $L^2_K$-eigenvalue is at least one, so it suffices to prove a lower bound on the multiplicity of $\mu_\omega(z)$ for those $\omega \in \UU$ such that the upper bound computed in (1) is strictly greater than 1. This argument will typically rely on the existence of symmetries $S$ of $H_z$ such that 
\begin{equation}
    S \big( L^2_{K,\omega} \big) = L^2_{K,\widetilde{\omega}}.
\end{equation}
This implies that $H_z$ on $L^2_{K,\lambda}$ and $L^2_{K,\lambda}$ are conjugated, hence isospectral: $\mu_\omega(z) = \mu_{\widetilde{\omega}}(z)$. In each case, this will provide a lower bound on the multiplicity of $\mu_\omega(z)$ as an $L^2_K$-eigenvalue equaling the upper bound computed in step (1), and thus we deduce that $\mu_\omega(z)$ has constant multiplicity, which we denote by $m_\omega$, for sufficiently small $z$. 

For such $z$, it then follows that $\mu_\omega(z)$ is equal to one of the eigenvalues of $H_z$ whose existence is guaranteed by Theorem \ref{thm:typeA}, and can thus be extended to an analytic function on $\mathbb{R}$, such that $\mu_\omega(z)$ is an eigenvalue of $H_z$ for all $z \in \mathbb{R}$. We can then apply Proposition \ref{prop:stronger} to conclude that $H_z$ has an $L^2_K$-eigenvalue $\mu_\omega(z)$ which has multiplicity $m_\omega$ for $z \in \mathbb{R}$ away from a discrete set $D_1$, and whose eigenprojector $\pi_\omega(z)$ is analytic on $\mathbb{R}$.

{\it (3) Computation of the characteristic polynomial:} Let 
\begin{equation}
    \UU_\omega = \{\widetilde{\omega} \in \UU \;:\;\mu'_{\widetilde{\omega}}(0) = \mu'_\omega(0) \},
\end{equation}
so that $m_\omega = |\UU_\omega|$ and $\mu_\omega = \mu_{\widetilde{\omega}}$ for all $\widetilde{\omega} \in \UU_\omega$ by steps (1) and (2). Then for all such $\widetilde{\omega}$ and sufficiently small $z$,
\begin{equation}
    \big(H_z - \mu_\omega(z)\pi_\omega(z)\big)|_{L^2_{K,\widetilde{\omega}}} = 0,
\end{equation}
and by analyticity this must hold for all $z \in U$. Therefore $\mu_\omega(z)$ is an $L^2_{K,\widetilde{\omega}}$-eigenvalue of multiplicity at least one for all $z \in U$, and thus is a simple $L^2_{K,\widetilde{\omega}}$-eigenvalue for all $z\in \mathbb{R}\setminus D_1$. As a result, for all such $z$ there exists a basis $(\phi_1,\ldots,\phi_{m_\omega})$, normalized to have $L^2_K$-norm 1, of the eigenspace $\cE$ corresponding to $\mu_\omega(z)$ consisting of precisely one vector from $L^2_{K,\widetilde{\omega}}$ for each $\widetilde{\omega} \in \UU_\omega$.

For $z \in \mathbb{R}\setminus D_1$, Lemma \ref{lem:perturbk} then describes the structure of the dispersion surfaces corresponding to $\mu_\omega(z)$ near the vertex $K$. For each of the lattices we examine, one of two things happens: either all of the eigenvalues of $M(z,\kappa) =-\pi_\omega(z)(2i\kappa\cdot \nabla)\pi_\omega(z)|_\cE$ are simple on an open set (not necessarily connected) near $\kappa = 0$, or $M(z,\kappa)$ is identically 0. In the first case, if $\lambda(z,\kappa)$ is a simple eigenvalue of $M(z,\kappa)$, then Lemma \ref{lem:perturbk}(2) tells us there exists a simple eigenvalue $\mu_\omega(z,\kappa)$ of $H_z$ on $L^2_{K+\kappa}$ such that
\begin{equation}\label{eq:holds}
    \mu_\omega(z,\kappa) = \mu_\omega(z) + \lambda(z,\kappa) + \cO(\|\kappa\|^2).
\end{equation}
Note that \eqref{eq:holds} also always holds at $\kappa = 0$, since $\lambda(z,0) = 0$ by virtue of $M(z,0) = 0$, although $\mu_\omega(z,\kappa)$ will typically no longer be simple at this point. As a result, in the specific case where the eigenvalues of $M(k)$ are simple on a punctured neighborhood of $\kappa = 0$, the \eqref{eq:holds} in fact holds on a neighborhood of $\kappa = 0$.  On the other hand, if $M(z,\kappa)$ is identically zero, then Lemma \ref{lem:perturbk}(3) tells us that every dispersion surface corresponding to $\mu_\omega(z)$ near the vertex $K$ satisfies 
\begin{equation}
    \mu_\omega(z,\kappa) = \mu_\omega(z)  + \cO(\|\kappa\|^2),
\end{equation}
which immediately implies that $(K, \mu(z))$ is a quadratic point (as per Definition \ref{def:1}).

Using the basis $(\phi_1,\ldots,\phi_{m_\omega})$, we can then compute the entries of $M(z,\kappa)$ with respect to this basis using Lemma \ref{lem:mkentries}. In particular, this lemma tells us that the diagonal entries of $M(z,\kappa)$ are all zero, and we only need to compute the entries above the diagonal since $M(z,\kappa)$ is Hermitian. Once we have an explicit expression for $M(z,\kappa)$, we can then compute its characteristic polynomial. 

We then finish by checking which of the coefficients of this polynomial are nonzero for $z$ sufficiently small, which by Lemma \ref{lem:mkanalytic} will then imply that these coefficients remain nonzero for all $z \in \mathbb{R}$ away from a discrete set $D_2$ by analyticity. To perform this computation, we will typically use the fact that a normalized eigenvector $\phi_{\omega}(x;z)$ corresponding to the $L^2_{K,\omega}$-eigenvalue $\mu_\omega(z)$ satisfies
\begin{equation}\label{eq:description}
    \phi_{\omega}(x;z) = \phi_{\omega}(x) + \cO(|z|),
\end{equation}
where $\phi_{\omega}(x)$ is the normalized eigenvector corresponding to $\mu_\omega(0)$ given by \eqref{eq:evectors}. This follows from the observation that $\pi_\omega(z)\phi_\omega$ is an eigenvector corresponding to $\mu_\omega(z)$ for $z$ sufficiently small, and the fact that $ \pi_\omega(z) = \pi_\omega(0)+\cO(|z|)$ by a Neumann series argument. Letting $D = D_1 \cup D_2$, we then conclude that  \eqref{eq:description} will hold for all $z \in \mathbb{R}$ away from the discrete set $D$.

\subsection{Proof of Theorem \ref{thm:2} for the Simple Cubic}\label{sec:4.3}

Let $\Lambda = \Lambda^S$, and let $(k_1,k_2,k_3)$, $\cB$, $K$, $m$, $G_0$, and $\UU$ be the objects listed in corresponding column (i.e. the first column) of Table \ref{tab:geometry}. We also let $V$ be a $\Lambda$-invariant potential and let $H_z = -\Delta+z V$. Lastly we will need the group elements $r,s,f$ defined in \eqref{eq:generators} and $f_1,f_2,f_3$ defined in \eqref{eq:elements1}.

{\it (1) Upper bound on multiplicity:} For each $\omega \in \UU$, $\mu'_\omega(0)$ is given by:
\begin{equation}
    \mu'_\omega = \sum_{j \in \JJ}\omega^jV_{m(j)}.
\end{equation}
A quick computation shows that, for $j=1,2,3$, 
\begin{equation}
    f_j^{-1}K = K - k_j.
\end{equation}
It follows from the definition of $m(j)$ (given by \eqref{eq:mj}) that $m(j)=-j$. Moreover, since $V$ is even (as noted in Section \ref{sec:3.2}), it follows that $V_{-j} =V_j$. Thus we have the formula 
\begin{equation}
    \mu'_\omega = \sum_{j \in \JJ}\omega^jV_{j}.
\end{equation}

In addition, observe that $V$ is invariant under $r$, which permutes the coordinate axes. Consequently, we also have the identities
\begin{equation}
    V_{1,0,0} = V_{0,1,0} = V_{0,0,1}\quad \text{ and } \quad V_{1,1,0} = V_{1,0,1} = V_{0,1,1}.
\end{equation}
It follows that we can rewrite $\mu'_\omega(0)$ as  
\begin{equation}
    \mu'_\omega(0) = V_{0,0,0} +  (\omega_1+\omega_2+\omega_3)V_{1,0,0}+ (\omega_2\omega_3+\omega_1\omega_3+\omega_1\omega_2)V_{1,1,0}  + \omega_1\omega_2\omega_3V_{1,1,1}.
\end{equation}

We then plug  $\omega$ into this formula for each $\omega \in \UU$, which gives us the following:
\begin{align}
    \mu'_{1,1,1}(0) & = V_{0,0,0}+3V_{1,0,0}+3V_{1,1,0}+V_{1,1,1},\\
    \mu'_{-1,1,1}(0) = \mu'_{1,-1,1}(0) = \mu'_{1,1,-1}(0) & = V_{0,0,0}+V_{1,0,0}-V_{1,1,0}-V_{1,1,1}, \\
    \mu'_{1,-1,-1}(0) = \mu'_{-1,1,-1}(0) = \mu'_{-1,-1,1}(0) & = V_{0,0,0}-V_{1,0,0}-V_{1,1,0}+V_{1,1,1},\\
    \mu'_{-1,-1,-1}(0) & = V_{0,0,0}-3V_{1,0,0}+3V_{1,1,0}-V_{1,1,1}.
\end{align}
Note that the set where the right-hand sides of any pair of the above 4 equations are equal describes a hyperplane. Consequently, the set where the right-hand sides of the above four equations fail to be distinct is a union of six hyperplanes. It follows that for $V$ away from a set of codimension 1, the eigenvalue $\|K\|^2$ of $-\Delta$ splits into at least two simple eigenvalues and two eigenvalues of multiplicity at most three.

{\it (2) Lower bound on multiplicity:} Observe that $f_1r = rf_3$, $f_2r = rf_1$, and $f_3r=rf_2$. As a result, if $\phi \in L^2_{K,(-1,1,1)}$ is an eigenvector of $H_z$, then $r_*\phi$ is also an eigenvector of $H_z$ with the same eigenvalue and 
\begin{align}
    (f_1)_*(r_*\phi) &= r_*(f_3)_*\phi = r_*\phi \\
    (f_2)_*(r_*\phi) &= r_*(f_1)_*\phi = -r_*\phi \\
    (f_3)_*(r_*\phi) &= r_*(f_2)_*\phi = r_*\phi.
\end{align}
Hence, $r_*\phi \in  L^2_{K,(1,-1,1)}$, and an identical computation shows that $r^2_*\phi$ is an eigenvector of $H_z$ as well, but in $L^2_{K,(1,1,-1)}$. Therefore $\mu_{(-1,1,1)}(z)$ is an $L^2_K$-eigenvalue with multiplicity at least 3, which together with step (1) implies that its multiplicity is exactly 3. The same argument applied to an eigenvector $\phi$ of $H_z$ in $L^2_{K,(1,-1,-1)}$ shows that $\mu_{(-1,1,1)}(z)$ is an $L^2_K$-eigenvalue with multiplicity 3 as well. Therefore, $H_z$ has two triple $L^2_K$-eigenvalues for all $z \in \mathbb{R}$ away from a discrete set $D_1$, and the corresponding eigenprojectors are analytic on $\mathbb{R}$.

{\it (3) Computation of the characteristic polynomial:} Fix some $z \in \mathbb{R}\setminus D_1$, and let $\phi_1 \in L^2_{K,(-1,1,1)}, \phi_2 \in L^2_{K,(1,-1,1)}$ and $\phi_3 \in L^2_{K,(1,1,-1)}$ be normalized eigenvectors for the eigenvalue $\mu_{(-1,1,1)}(z)$ of $H_z$. The entries of $M(z,\kappa)$ with respect to this basis are given by $-2i\kappa\cdot \langle \phi_j,\nabla \phi_\ell\rangle$. For $j \not = \ell$, $\langle \phi_j,\nabla \phi_\ell\rangle$ is an eigenvector of both $f_j$ and $f_\ell$ with eigenvalue -1 by Lemma \ref{lem:mkentries}, and thus it lies in $\CC e_j \cap \CC e_\ell = \{0\}$. It follows that $M(z,\kappa) = 0$ for all $z$ and $\kappa$, and thus we conclude that 
\begin{equation}
    \mu_{(-1,1,1)}(z,\kappa) = \mu_{(-1,1,1)}(z)+ \cO(|\kappa|^2).
\end{equation}
By Definition \ref{def:1} this means that $(K,\mu_{(-1,1,1)}(z))$ is a 3-fold quadratic point for all $z \in \mathbb{R}\setminus D_1$. The exact same argument shows that $(K,\mu_{(-1,1,1)}(z))$ is a 3-fold quadratic point as well. This completes the proof of Theorem \ref{thm:2} when $\Lambda$ is a simple cubic lattice.

\subsection{Proof of Theorem \ref{thm:2} for the Body-Centered Cubic at \texorpdfstring{$K = (\pi,\pi,\pi)$}{K=(pi,pi,pi)}}\label{sec:4.4}

Let $\Lambda = \Lambda^{BC}$, let $K = (\pi,\pi,\pi)$ and let $(k_1,k_2,k_3)$, $\cB$, $m$, $G_0$, and $\UU$ be the objects listed in the corresponding column of Table \ref{tab:geometry} (i.e. the second of the three columns for the first three rows and the second of the four columns for the remaining rows). We also let $V$ be a $\Lambda$-invariant potential and let $H_z = -\Delta+z V$. Lastly we will again need the generators $r,s,f$, and also the group elements $f_{12},f_{13},f_{23}$ defined in \eqref{eq:elements2}.

{\it (1) Upper bound on multiplicity:} Just as we did in \S \ref{sec:4.3}, we start by computing relations among the Fourier coefficients $V_{m(j)}$ for $j \in \JJ$. To begin, we compute that
\begin{align}
    f_{13}^{-1}K & = K -k_1-k_3,\\
    f_{23}^{-1}K & = K-k_2-k_3,\\
    f_{12}^{-1}K & = K-k_1-k_2-k_3.
\end{align}
Again using the fact that $V$ is even, it follows that for $\omega \in \UU$, 
\begin{equation}
    \mu'_\omega(0) = V_{0,0,0}+\omega_1V_{1,0,1} + \omega_2V_{0,1,1} + \omega_1\omega_2V_{1,1,1}.
\end{equation}

Also note that $V$ being invariant under $r$ implies $V_{1,0,1} = V_{0,1,1}$. In addition, $V_{1,1,1} =V_{1,0,0}$ since
\begin{equation}
    V_{k_1 + k_3} = r_*V_{k_1+ k_3} =V_{r^\top(k_1+ k_3)} =V_{k_2+ k_3}.
\end{equation}
Therefore 
\begin{equation}
    \mu'_\omega(0) = V_{0,0,0}+(\omega_1+ \omega_2 + \omega_1\omega_2)V_{1,1,1}.
\end{equation}

We then plug $\omega$ into this formula for each $\omega \in \UU$ to obtain:
\begin{align}
    \mu'_{1,1}(0) & = V_{0,0,0}+3V_{1,1,1}, \\
    \mu'_{-1,1}(0) = \mu'_{1,-1}(0) =\mu'_{-1,-1}(0) & = V_{0,0,0}-V_{1,1,1}.
\end{align}
The set where the right-hand sides of the above two equations fail to be distinct is a single hyperplane. Therefore we again conclude that for $V$ away from a set of codimension 1, the eigenvalue $\|K\|^2$ of $-\Delta$ splits into at least a simple eigenvalue and an eigenvalue of multiplicity at most three.

{\it (2) Lower bound on multiplicity:} Just as in step (2) of Section \ref{sec:4.3}, observe that $f_{13}r = rf_{23}$ and $f_{23}r = rf_{12}$. As a result, if $\phi$ is an eigenvector of $H_z$ in $L^2_{K,(-1,1)}$, then $r_*\phi$ and $r^2_*\phi$ are again eigenvectors of $H_z$ with the same eigenvalue on $L^2_{K,(1,-1)}$ and $L^2_{K,(-1,-1)}$, respectively. Therefore $H_z$ has a triple $L^2_K$-eigenvalue for all $z\in \mathbb{R}$ away from a discrete set $D_1$, and the corresponding eigenprojector is analytic on $\mathbb{R}$.

{\it (3) Computation of the characteristic polynomial:} Fix some $z \in \mathbb{R}\setminus D_1$, and let $\phi_1 \in  L^2_{K,(-1,1)}$ be a normalized eigenvector of $H_z$ for the eigenvalue $\mu_{(-1,1)}(z)$. Then, as we saw in step (2), $\phi_2 := r_*\phi \in L^2_{K,(1,-1)}$ and $\phi_3 := r_*^2\phi \in L^2_{K,(-1,-1)}$ are eigenvectors of $H_z$ with eigenvalue $\mu_{(-1,1)}(z)$ as well, and thus form a basis for the corresponding eigenspace. The entries of $M(z,\kappa)$ with respect to this basis are given by $-2i\kappa\cdot \langle \phi_j,\nabla \phi_\ell\rangle$, and we also note that $\langle \phi_2,\nabla \phi_3\rangle = \langle r_*\phi_1,\nabla (r_*\phi_1)\rangle = r\langle \phi_1,\nabla \phi_2\rangle$. Therefore, by Lemma \ref{lem:mkentries}, the entries of $M(z,\kappa)$ are entirely determined by $\langle \phi_1,\nabla \phi_2\rangle$ and $\langle \phi_1,\nabla \phi_3\rangle$.

To compute these, note that $f_{12}=f_{13}f_{23}$, and so again by Lemma \ref{lem:mkentries},
\begin{equation}
    f_{12}\langle \phi_1,\nabla \phi_2\rangle = (-1)^2\langle \phi_1,\nabla \phi_2\rangle = \langle \phi_1,\nabla \phi_2\rangle.
\end{equation}
Hence, $\langle \phi_1,\nabla \phi_2\rangle$ is an eigenvector of $f_{12}$ with eigenvalue 1, and is therefore of the form $\alpha e_3$ for some $\alpha \in \CC$. An identical argument applied to $\langle \phi_1,\nabla \phi_3\rangle$ and the element $f_{13} \in G$ implies that $\langle \phi_1,\nabla \phi_3\rangle = \beta e_2$ for some $\beta \in \CC$.

Note that $\alpha = -\overline{\beta}$:
\begin{align}
    \alpha = e_3\cdot \langle \phi_1,\nabla \phi_2\rangle &  = re_2\cdot \langle \phi_1,\nabla \phi_2\rangle = e_2\cdot r^2_*\langle r^2_* \phi_1,r^2_*\nabla \phi_2\rangle \\
    & = e_2 \cdot \langle \phi_3,\nabla \phi_1\rangle = -e_2 \cdot \overline{\langle \phi_1,\nabla \phi_3\rangle} = -\overline{\beta}.
\end{align}

Thus, with respect to the basis $\phi_1,\phi_2,\phi_3$,  $M(z,\kappa)$ is given by: 
\begin{equation}
    M(z,\kappa) = -2i \begin{pmatrix}[1.5] 0 & \alpha\kappa_3 & -\overline{\alpha}\kappa_2 \\
    -\overline{\alpha}\kappa_3 & 0 & \alpha\kappa_1 \\
    \alpha\kappa_2 & -\overline{\alpha}\kappa_1 & 0 \\
    \end{pmatrix}.
\end{equation}
A quick computation then gives the characteristic polynomial of $M(z,\kappa)$ (as a polynomial in $\mu$):
\begin{equation}
    \mu^3 - 4|\alpha|^2\|\kappa\|^2\mu+16\Im(\alpha^3)\kappa_1\kappa_2\kappa_3.
\end{equation}
It follows that the eigenvalues of $M(z,\kappa)$ will be simple away from $\kappa = 0$ as long as the coefficients $|\alpha|^2$ and $\Im (\alpha^3)$ are nonzero. 

By Lemma $\ref{lem:mkanalytic}$, the coefficients $|\alpha|^2$ and $\Im (\alpha^3)$ are analytic in $z$, and therefore will be nonzero away from a discrete set if they are nonzero for $z$ sufficiently small. However, by Lemma \ref{lem:evectors} we can assume that, for $z$ sufficiently small, $\phi_1,\phi_2$ are given by:
\begin{align}
    \phi_1(x;z) & = \frac{1}{2}\left(e^{iK\cdot x} - e^{if_{13}K\cdot x}+ e^{if_{23}K\cdot x}  - e^{if_{12}K\cdot x}\right)+\cO(|z|),\\
    \phi_2(x;z) & =\frac{1}{2}\left(
    e^{iK\cdot x} + e^{if_{13}K\cdot x}- e^{if_{23}K\cdot x}  - e^{if_{12}K\cdot x}\right)+\cO(|z|).
\end{align}
It follows that, for $z$ small, 
\begin{align}
\alpha = 
    e_3\cdot \langle \phi_1, \nabla \phi_2 \rangle & = \frac{i}{4}e_3\cdot \left(K - f_{13}K - f_{13}K +f_{12}K \right) +\cO(|z|^2) = \pi i +\cO(|z|^2).
\end{align}
Therefore both $|\alpha|^2$ and $\Im (\alpha^3)$ are nonzero for $z$ sufficiently small, and thus remain nonzero for all $z \in U$ away from another discrete set $D_2$. It follows that they are nonzero on $\mathbb{R} \setminus D_2$, and so by Definition \ref{def:1}, we conclude that $(K,\mu_{(-1,1)}(z))$ is a 3-fold Weyl point for all  $z \in \mathbb{R}\setminus (D_1 \cup D_2)$.

\subsection{Proof of Theorem \ref{thm:2} for the Body-Centered Cubic at \texorpdfstring{$K = (0,0,2\pi)$}{K=(0,0,2pi)}}\label{sec:4.5}

Let $\Lambda = \Lambda^{BC}$, let $K = (0,0,2\pi)$ and let $(k_1,k_2,k_3)$, $\cB$, $m$, $G_0$, and $\UU$ be the objects listed in the corresponding column of Table \ref{tab:geometry} (i.e. the second of the three columns for the first three rows and the third of the four columns for the remaining rows). We also let $V$ be a $\Lambda$-invariant potential and let $H_z = -\Delta+z V$. Lastly we will need the group element $s_0$ defined in \eqref{eq:elements3}.

{\it (1) Upper bound on multiplicity:} We start by computing relations among the Fourier coefficients $V_{m(j)}$ for $j \in \JJ$. In particular, we compute that
\begin{align}
    r^{-1}K  & = K + k_2, &  r^{-1}fK & = K - k_2 - k_3, \\
    r^{-2}K  & = K + k_1, & r^{-2}fK & = K - k_1- k_3 \\
    f^{-1}K  & = K - k_3. & & 
\end{align}
Again using the fact that $V$ is even, it follows that for $\omega \in \UU$, 
\begin{equation}
    \mu'_\omega(0) = V_{0,0,0}+\omega_1V_{0,1,0} +\omega_1^2V_{1,0,0} + \omega_2V_{0,0,1} + \omega_1\omega_2V_{0,1,1} + \omega_1^2\omega_2V_{1,0,1}.
\end{equation}
We also have that
\begin{equation}
    V_{k_1} = (f_1)_* V_{k_1} =V_{f_1^\top k_1} =V_{k_1 + k_3},
\end{equation}
which tells us that $V_{1,0,0} = V_{1,0,1}$. Furthermore, since $V$ is invariant under $r$, we obtain $V_{1,0,0} = V_{0,1,0}= V_{1,0,1} =  V_{0,1,1}$. Thus we can rewrite $\mu'_\omega(0)$ as
\begin{equation}
    \mu'_\omega(0) = V_{0,0,0}+\big(\omega_1 +\omega_1^2  + \omega_1\omega_2 + \omega_1^2\omega_2\big)V_{1,0,0}+ \omega_2V_{0,0,1}.
\end{equation}

We then plug $\omega$ into this formula for each $\omega \in \UU$ to obtain:
\begin{align}
    \mu'_{1,1}(0) & = V_{0,0,0}+4V_{1,0,0}+V_{0,0,1},\\
    \mu'_{\zeta_3,1}(0) = \mu'_{\overline{\zeta_3},1}(0)  & = V_{0,0,0}-2V_{1,0,0}+V_{0,0,1},\\
    \mu'_{1,-1}(0) =\mu'_{\zeta_3,-1}(0) = \mu'_{\overline{\zeta_3},-1}(0)  & = V_{0,0,0}-V_{0,0,1}.
\end{align}
The set where the right-hand sides of the above three equations fail to be distinct is a union of three hyperplanes. Therefore we conclude that for $V$ away from a set of codimension 1, the eigenvalue $\|K\|^2$ of $-\Delta$ splits into at least a simple eigenvalue, an eigenvalue of multiplicity at most two, and an eigenvalue of multiplicity at most three.

{\it (2) Lower bound on multiplicity:} Let $T$ be the conjugate-parity operator: $Tf(x) = \overline{f(-x)}$, and let $\phi_1$ be a normalized eigenvector of $H_z$ in $L^2_{K,(\zeta_3,1)}$ for $\mu_{(\zeta_3,1)}(z)$. For $z \in \mathbb{R}$, $\phi_2 = T\phi_1 \in L^2_K$ is also an eigenvector of $H_z$ for the same eigenvalue since $V$ is even and real. In addition, observe that
\begin{align}
    r_* \phi_2(x) & = \overline{\phi_1(-r^\top x)} =   \overline{\zeta_3 \phi_1(-x)} = \overline{\zeta_3} \phi_2(x) \\
    f_* \phi_2(x) & = \overline{\phi_1(-f^\top x)} = \overline{\phi_1(-x)} = \phi_2(x).
\end{align}
Therefore $\phi_2 \in L^2_{K,(\overline{\zeta_3},1)}$.

To give a lower bound on the multiplicity of $\mu_{(1,-1)}(z)$, let $L^2_{K,-1} = \ker_{L^2_K}(f_*+1)$, i.e. the space of odd functions in $L^2_K$. By construction of the subspaces $L^2_{K,\omega}$, it follows that 
\begin{equation}
    L^2_{K,-1} = L^2_{K,(1,-1)}\oplus L^2_{K,(\zeta_3,-1)} \oplus L^2_{K,(\overline{\zeta_3},-1)}.
\end{equation}
Also note that $s_0K = K-k_3$, and so $s_0$ is $K$-invariant, which together with the fact that $s_0$ commutes with $f$, implies that $(s_0)_*$ is a well-defined operator on $L^2_{K,-1}$.

Now let 
\begin{align}
    \psi_1(x) & =  \sin(2\pi x_1) + i\sin (2\pi x_2), \\
    \psi_2(x) & =  \sin(2\pi x_1) - i\sin (2\pi x_2), \\
    \psi_3(x) & = \sqrt{2}\sin(2\pi x_3).
\end{align}
A quick computation confirms that $\psi_j \in L^2_{K,-1}$, $\|\psi_j\| = 1$, and $-\Delta \psi_j = (2\pi)^2\psi_j = \|K\|^2\psi_j$ for $j=1,2,3$. In addition, note that $\sigma((s_0)_*) = U_4$, the fourth roots of unity. If we let $\cE_\omega = L^2_{K,-1} \cap\ker_{L^2_{K,-1}}((s_0)_* - \omega)$ for $\omega \in U_4$, then $\psi_1 \in \cE_{-i}$, $\psi_2 \in \cE_i$ and $\phi_3 \in \cE_{-1}$. Therefore $\|K\|^2$ is a simple eigenvalue of $-\Delta$ on $\cE_\omega$ for $\omega \in \{ -i,i,-1\}$, and so by Corollary \ref{cor:simple} it follows that, for sufficiently small $z \in \mathbb{R}$, there is a unique eigenvalue $\lambda(z)$ of $H_z$ on $\cE_{-i}$ satisfying $\lambda(z) = \|K\|^2 + \cO(|z|)$. Let $\Psi_z \in \cE_{-i}$ denote the normalized eigenvector corresponding to $\lambda_z$, and let $\Phi_z \in L^2_{K,(1,-1)}$ denote a normalized eigenvector corresponding to $\mu_{(1,-1)}(z)$, so that 
\begin{equation}\label{eq:evectorperturb}
    \Psi_z = \psi_1 + \cO(|z|),\quad \text{ and } \quad \Phi_z = \phi+\cO(|z|),
\end{equation}
where $\phi$ is defined by \eqref{eq:evectors} with $\omega = (\zeta_3,-1)$.

Now assume for contradiction that $\mu_{(1,-1)}(z)$ has multiplicity strictly less than 3 for $z \not = 0$. Then $T\Psi_z$ and $T\Phi_z$ are also eigenvectors corresponding to $\lambda(z)$ and $\mu_{(1,-1)}$, respectively, and so both of these eigenvalues must have multiplicity at least 2. Since we are assuming that the multiplicity of $\mu_{(1,-1)}$ is strictly less than 3, we deduce that these eigenvalues must in fact be equal, and their multiplicity is exactly 2. 

As a result, for all $z\in \mathbb{R}$, nonzero and sufficiently small,
$\spa(\Psi_z,T\Psi_z) = \spa(\Phi_z,T\Phi_z)$. Therefore we can express $\Phi_z$ with respect to $\Psi_z$ and $T\Psi_z$ as 
\begin{equation}\label{eq:zlimit}
    \Phi_z = \langle \Phi_z,\Psi_z\rangle \Psi_z + \langle \Phi_z,T\Psi_z\rangle T\Psi_z,
\end{equation}
where we have used the fact that $T\Psi_z \in \cE_i$, and is therefore orthogonal to $\Psi_z$. Taking the limit of both sides of \eqref{eq:zlimit} as $z \rightarrow 0$, we obtain
\begin{equation}
    \phi = \langle \phi,\psi_1\rangle \psi_1 + \langle \phi,T\psi_1\rangle T\psi_1.
\end{equation}

This is not possible: by \eqref{eq:evectors}, the left-hand side depends on $x_3$, while the right-hand side depends only on $x_1,x_2$. We conclude that $H_z$ has a double and a triple $L^2_K$-eigenvalue for all $z\in \mathbb{R}$ away from a discrete set $D_1$, and the corresponding eigenprojector is analytic on $\mathbb{R}$.

{\it (3) Computation of the characteristic polynomial:} Fix some $z \in U\cap \mathbb{R}$, and let $\phi_1 \in  L^2_{K,(\zeta_3,1)}$, $ \phi_2 \in L^2_{K,(\overline{\zeta_3},1)}$ be normalized eigenvectors for the eigenvalue $\mu_{(\zeta_3,1)}(z)$ of $H_z$. The entries of $M(z,\kappa)$ with respect to this basis are given by $-2i\kappa\cdot \langle \phi_j,\nabla \phi_\ell\rangle$, and by Lemma \ref{lem:mkentries}, the entries of $M(z,\kappa)$ are entirely determined by $\langle \phi_1,\nabla \phi_2\rangle$. However, this same lemma also tells us that $\langle \phi_1,\nabla \phi_2\rangle$ is an eigenvector of $f$ with eigenvalue 1, and therefore must be the zero vector. It follows that $M(z,\kappa)=0$ for all $z$ and $\kappa$, and thus we conclude that $(K,\mu_{(\zeta,1)}(z)$ is a 2-fold quadratic point for all $z \in \mathbb{R}\setminus D_1$. 

Now, let $\phi_1 \in  L^2_{K,(1,-1)}$, $\phi_2 \in  L^2_{K,(\zeta_3,-1)}$, and $ \phi_3 \in L^2_{K,(\overline{\zeta_3},-1)}$ be normalized eigenvectors for the eigenvalue $\mu_{(1,-1)}(z)$ of $H_z$. Then the same argument implies that, for $j,\ell \in \{1,2,3\}$, $j \not \ell$, $\langle \phi_1,\nabla \phi_2\rangle$ is again an eigenvector of $f$ with eigenvalue 1 and therefore must be the zero vector. We thus conclude that $(K,\mu_{(1,-1)}(z)$ is a 3-fold quadratic point for all $z \in \mathbb{R}\setminus D_1$.

\subsection{Proof of Theorem \ref{thm:2} for the Face-Centered Cubic}\label{sec:4.6}

Let $\Lambda = \Lambda^{FC}$, and let $(k_1,k_2,k_3)$, $\cB$, $K$, $m$, $G_0$, and $\UU$ be the corresponding objects listed in the final column of Table \ref{tab:geometry}. We also let $V$ be a $\Lambda$-invariant potential and let $H_z = -\Delta+z V$. Lastly we will again need the group element $s_0$ defined in \eqref{eq:elements3}.

{\it (1) Upper bound on multiplicity:} We start by computing relations among the Fourier coefficients $V_{m(j)}$ for $j \in \JJ$. In particular, we compute that
\begin{align}
    s_0^{-1}K & = K - k_1,\\
    s_0^{-2}K & = K - k_1 -k_2 +k_3,\\
    s_0^{-3}K & = K - k_2.
\end{align}
Again using the fact that $V$ is even and invariant under $r$, it follows that for $\omega \in \UU$, 
\begin{equation}
    \mu'_\omega(0) = V_{0,0,0}+\omega V_{1,0,0} + \omega^2 V_{1,1,-1} + \omega^3V_{0,1,0} = V_{0,0,0}+(\omega+\omega^3) V_{1,0,0} + \omega^2 V_{1,1,-1}.
\end{equation}
We then plug $\omega$ into this formula for each $\omega \in \UU$ to obtain:
\begin{align}
    \mu'_1(0) & = V_{0,0,0}+2 V_{1,0,0} +  V_{1,1,-1},\\
    \mu'_i(0) = \mu'_{-i}(0)  & = V_{0,0,0}-  V_{1,1,-1},\\
    \mu'_{-1}(0)  & =V_{0,0,0}-2 V_{1,0,0} +  V_{1,1,-1}.
\end{align}
The set where the right-hand sides of the above three equations fail to be distinct is a union of three hyperplanes. Therefore we conclude that for $V$ away from a set of codimension 1, the eigenvalue $\|K\|^2$ of $-\Delta$ splits into at least two simple eigenvalues and an eigenvalue of multiplicity at most two.

{\it (2) Lower bound on multiplicity:} Again let $T$ be the conjugate-parity operator and let $\phi_1$ be a normalized eigenvector of $H_z$ in $L^2_{K,i}$. For $z \in \mathbb{R}$, $\phi_2 = T\phi_1 \in L^2_K$ is also an eigenvector of $H_z$ for the same eigenvalue, and
\begin{equation}
    (s_0)_* \phi_2(x) = \overline{\phi_1(-s_0^\top x)}  = \overline{i\phi_1(- x)}  = -i \phi_1(x),
\end{equation}
which implies $\phi_2 \in L^2_{K,-i}$. Therefore $H_z$ has a double $L^2_K$-eigenvalue for all $z\in \mathbb{R}$ away from a discrete set $D_1$, and the corresponding eigenprojector is analytic on $\mathbb{R}$.

{\it (3) Computation of the characteristic polynomial:} Fix some $z \in \mathbb{R} \setminus D_1$, and let $\phi_1 \in  L^2_{K,i}$, $\phi_2 \in  L^2_{K,-i}$ be normalized eigenvectors for the eigenvalue $\mu_i(z)$ of $H_z$. The entries of $M(z,\kappa)$ with respect to this basis are given by $-2i\kappa\cdot \langle \phi_j,\nabla \phi_\ell\rangle$, and by Lemma \ref{lem:mkentries}, the entries of $M(z,\kappa)$ are entirely determined by $\langle \phi_1,\nabla \phi_2\rangle$. This same lemma also tells us that $\langle \phi_1,\nabla \phi_2\rangle$ is an eigenvector of $s_0$ with eigenvalue -1, and is therefore of the form $\alpha e_3$ for some $\alpha \in \CC$.

Thus, with respect to the basis $\phi_1,\phi_2$, $M(z,\kappa)$ is given by 
\begin{equation}
    M(\kappa) = -2i \begin{pmatrix} 0 & \alpha\kappa_3  \\
    -\overline{\alpha}\kappa_3 & 0  \\
    \end{pmatrix},
\end{equation}
and its characteristic polynomial is  $\mu^2 - 4|\alpha|^2\kappa_3^2$. It follows that the eigenvalues of $M(z,\kappa)$ can be written as $\lambda(z,\kappa) = \pm 2|\alpha \kappa_3|$, and therefore these eigenvalues will be simple for $\kappa_3 \not =0$ as long as $\alpha \not = 0$. 

By Lemma $\ref{lem:mkanalytic}$, the coefficient $|\alpha|^2$ is analytic in $z$, and therefore will be nonzero away from a discrete set if it is nonzero for $z$ sufficiently small. However, by Lemma \ref{lem:evectors} we can assume that, for $z$ sufficiently small, $\phi_1,\phi_2$ are given by:
\begin{align}
    \phi_1(x) & = \frac{1}{2}\left(e^{iK\cdot x} +i e^{iS_0^3K\cdot x}- e^{is_0^2K\cdot x}  -i e^{is_0K\cdot x}\right) +\cO(|z|), \\
    \phi_2(x) & = \frac{1}{2}\left(e^{iK\cdot x} -i e^{is_0^3K\cdot x}- e^{is_0^2K\cdot x}  +i e^{is_0K\cdot x}\right)+\cO(|z|).
\end{align}
It follows that, for $z$ small, 
\begin{align}
\alpha = 
    e_3\cdot \langle \phi_1, \nabla \phi_2 \rangle & = \frac{i}{4}e_3\cdot\left(K - s_0^3K +s_0^2K-s_0K \right)+\cO(|z|^2) = \pi i+\cO(|z|^2).
\end{align}
Therefore $|\alpha|^2$ is nonzero for $z$ sufficiently small, and thus remains nonzero on $U$ away from another discrete set $D_2$. It follows that $|\alpha|^2$ is nonzero on $\mathbb{R}\setminus D_2$, and so by Definition \ref{def:1} we conclude that $(K,\mu_i(z))$ is a basin point for all $z \in \mathbb{R}\setminus (D_1 \cup D_2)$.

\begin{appendix}

\section{Appendix: Spectral Theory of the Laplacian on \texorpdfstring{$L^2_K$}{L2K}}\label{app:A}
In this appendix, we compute the spectrum of of $-\Delta$ on $L^2_K$ and show that the cardinality of the set $[K]$, defined by \eqref{eq:equivclass}, is equal to the multiplicity of $\|K\|^2$ as an eigenvalue. We then use this to compute some bounds on the multiplicity of $\|K\|^2$ when $K \in \cB$, and lastly show that if $K$ is a vertex of $\cB$, then $[K]$ is a subset of the vertices of $\cB$.

Fix some lattice $\Lambda$ with basis $v_1,\ldots, v_n$ and reciprocal basis $k_1,\ldots, k_n$, and fix some $K \in \mathbb{R}^n$. Just as we did following \eqref{eq:Fouriercoef}, we also let $mk = m_1k_1 + \cdots m_nk_n$ for $m \in \Z^n$. We then claim that $\phi_m(x) = e^{i(K+mk)\cdot x}$ for $m \in \Z^n$ is an orthonormal basis of eigenvectors for $-\Delta$ on $L^2_K$. Indeed, note that 
\begin{equation}
    -\Delta \phi_m(x) = \|K+mk\|^2\phi_m(x),
\end{equation}
and $(\phi_m)_{m \in \Z^n}$ is the image of the orthonormal basis $(\otimes_{j=1}^ne^{2\pi i m_jx_j})_{m \in \Z^n}$ of $L^2[0,1]^{\otimes n}$ under the unitary map which first sends $\otimes_{j=1}^ne^{2\pi i m_jx_j}$ to $e^{imk\cdot x}\in L^2_0$, and then $e^{imk\cdot x}$ to $\phi_m$ via multiplication by $e^{iK\cdot x}$. Consequently, 
\begin{equation}
    \sigma(-\Delta) = \{\|K+mk\|^2\;:\; m \in \Z^n\}.
\end{equation}
and the multiplicity of an eigenvalue $\mu_m := \|K+mk\|^2$ is given by 
\begin{equation}\label{eq:laplacianmult}
    m_{-\Delta}(\mu_m) = \left|\left\{k' \in K+\Lambda^*\;:\; \|k'\|^2 = \|K+mk\|^2\right\}\right|.
\end{equation}
In particular, this proves \eqref{eq:equivclass}; namely that the the cardinality of the set $[K]$ is equal to the multiplicity of $\mu_0 = \|K\|^2$.

From here, recall that the Floquet--Bloch problem \eqref{eq:floquet} is periodic with respect to the dual lattice $\Lambda^*$, and so we focus our analysis on $K \in \cB$. For such $K$, the minimal eigenvalue of $-\Delta$ on $L^2_{K}$ is then given by $\mu_0$, since by definition of the Brillouin zone $\cB$,
\begin{equation}\label{eq:ineq}
    \|K\|^2 \leq \|K-k'\|^2, \quad \forall k' \in \Lambda^*.
\end{equation}
This also implies that if $K$ is in the interior of $\cB$, then the inequality \eqref{eq:ineq} is in fact strict, and so the eigenvalue $\mu_0$ is simple. Conversely, we expect $K \in \pd \cB$, and in particular the vertices of $\cB$, to correspond to eigenvalues of high multiplicity, as the following proposition demonstrates.

\begin{proposition}\label{prop:brillouinmult}
Let $K\in \cB$, let $\mu_0 = \|K\|^2$, and let $m$ be the number of (hyper)faces of $\cB$ which contain $K$, where $m$ is possibly zero. Then there exist vectors $K_1,\ldots,K_m \in \Lambda^*$ such that $K- K_j$ also lies on $m$ (hyper)faces of $\cB$ and $\|K - K_j\|^2 = \|K\|^2$ for $j=1,\ldots,m$, so that $m_{-\Delta}(\mu_0) \geq m+1$. Furthermore, $m_{-\Delta}(\mu_0) =1$ if and only if $m=0$ and $m_{-\Delta}(\mu_0) = 2$ if and only if $m=1$.
\end{proposition}
\begin{proof}
Let $K \in \cB$ such that $K$ lies on $m$ (hyper)faces of $\cB$ for some non-negative integer $m$. Then if $m > 0$, there exist vectors $K_1,\ldots,K_m \in \Lambda^*$ such that $K$ lies on the (hyper)planes defined by $x \cdot K_j  =  \tfrac{1}{2}\|K_j\|^2$ for $j=1,\ldots,m$, where each of these (hyper)planes intersected with $\cB$ is precisely one of the $m$ (hyper)faces containing $K$. Then for $j=1,\ldots,m$
\begin{align}\label{eq:brillouin1}
    \|K - K_j\|^2 & = \|K\|^2 - 2K\cdot  K_j + \|K_j\|^2 \\
    & = \|K\|^2 - \|K_j\|^2 + \|K_j\|^2\\
    & = \|K\|^2.
\end{align}
Therefore it follows from \eqref{eq:laplacianmult} that $\mu_0$ has multiplicity of at least $m+1$. 

To prove that $K - K_j$ lies on $m$ (hyper)faces of $\cB$, observe that since $K \in \cB$, it follows that,
\begin{equation}\label{eq:brillouin2}
    \|K -K_j\|^2 = \|K\|^2 \leq \|(K - K_j)-K\|^2, \quad \forall K \in \Lambda^*,
\end{equation}
which implies $K - K_j \in \cB$. Furthermore, we have that
\begin{equation}\label{eq:brillouin3}
    (K-K_j)\cdot (-K_j) = -K\cdot K_j+ \|K_j\|^2 = -\frac{1}{2}\|K_j\|^2 - \|K_j\|^2 = \frac{1}{2}\|K_j\|^2, 
\end{equation}
and for all $\ell \not = j$,
\begin{align}\label{eq:brillouin4}
    (K-K_j)\cdot (K_\ell-K_j) & = K\cdot K_\ell - K\cdot K_j - K_j\cdot K_\ell + \|K_j\|^2\\
    & = \frac{1}{2}\|K_\ell\|^2 - \frac{1}{2}\|K_j\|^2 - K_j\cdot K_\ell + \|K_j\|^2 \\
    & = \frac{1}{2}\left(\|K_\ell\|^2 - 2K_j\cdot K_\ell + \|K_j\|^2\right) \\
    & = \frac{1}{2}\|K_\ell-K_j\|^2.
\end{align}
Therefore $K-K_j$ lies on the $m$ (hyper)planes defined by $x \cdot(- K_j)  =  \tfrac{1}{2}\|K_j\|^2$ and $x \cdot (K_j-K_\ell)  =  \tfrac{1}{2}\|K_j-K_\ell\|^2$ for $\ell \not = j$.

We now seek to show that each of these (hyper)planes defines a (hyper)face of $\cB$. To start, for $j = 0,\ldots, m$ and $\ell = 1,\ldots,m$, let 
\begin{equation}
    P_{j\ell} = \begin{cases} \{x \in \mathbb{R}^n \;:\; x \cdot K_\ell  =  \tfrac{1}{2}\|K_\ell\|^2\} & j = 0 \\
    \{x \in \mathbb{R}^n \;:\; -x \cdot K_j  =  \tfrac{1}{2}\|K_j\|^2\} & j \not = 0, \ell = j \\
    \{x \in \mathbb{R}^n \;:\; x \cdot (K_\ell - K_j)  =  \tfrac{1}{2}\|K_j- K_\ell\|^2\} & j \not = 0, \ell \not = j, 
    \end{cases}
\end{equation}
and suppose $k' \in P_{0\ell} \cap \cB$ for some $\ell$. Then the same computations as in \eqref{eq:brillouin1}-\eqref{eq:brillouin2}, but with $K$ replaced with $k'$, imply that $k'-K_j \in \cB$ for $j=1,\ldots,m$. Similarly, \eqref{eq:brillouin3} with $K$ replaced with $k'$ implies $k' - K_j \in P_{jj}$, and \eqref{eq:brillouin4} with $K$ replaced with $k$ implies $k' - K_j \in P_{j\ell}$ for $j\not = \ell,0$. As a result, for $j=1,\ldots, m$,
\begin{equation}
     (P_{0\ell} \cap \cB) -K_j = P_{j\ell} \cap \cB.
\end{equation}
By construction though, $P_{0\ell} \cap \cB$ is a (hyper)face of $\cB$, and since $P_{j\ell} \cap \cB$ is an isometric set and must be contained in the boundary of $\cB$, it follows that $P_{j\ell} \cap \cB$ is in fact a (hyper)face of $\cB$ as well.

For the second part of the proposition statement, observe that it suffices to prove that $m_{-\Delta}(\mu_0) \leq m+1$ when $m=0,1$. However, if $m=0$ then this means that $K$ lies on zero (hyper)faces, and therefore must be in the interior of $\cB$. We have already seen that in this case the eigenvalue $\mu_0 = \|K_0\|^2$ is simple, and thus $m_{-\Delta}(\mu_0) =1$, as desired. Now assume that $m=1$, so that $K$ lies on a single (hyper)face of $\cB$, and let $m' =m_{-\Delta}(\mu_0)$. Then by again using \eqref{eq:laplacianmult}, we deduce that there exist vectors $K_1,\ldots,K_{m'} \in \Lambda^*$ such that $\|K - K_j\|^2 = \|K\|^2$ for $j=1,\ldots,m'$. As a result, \eqref{eq:brillouin4} implies that $K\cdot K_j  =  \tfrac{1}{2}\|K_j\|^2$ for $j=1,\ldots,m'$, and so $K$ lies on the $m'$ distinct (hyper)planes defined by $x \cdot K_j  =  \tfrac{1}{2}\|K_j\|^2$ for $j=1,\ldots,m'$. However, since $K$ lies on a single (hyper)face of $\cB$, this implies $K$ lies on exactly one of these (hyper)planes. Therefore $m' = m_{-\Delta}(\mu_0)=2$, as claimed.
\end{proof}

\begin{proposition}\label{cor:brillouinvert}
Let $V(\cB)$ denote the vertices of $\cB$ and let $K \in V(\cB)$. Then
\begin{equation}
   [K] =  V(\cB) \cap (K_0 + \Lambda^*).
\end{equation}
\end{proposition}
\begin{proof}
Let $K \in V(\cB)$ and let $k' \in \Lambda^*$; it then suffices to prove that $\|K - k'\|^2 = \|K\|^2$ if and only if $K - k' \in V(\cB)$. However, \eqref{eq:brillouin1} implies that $\|K - k'\|^2 = \|K\|^2$ if and only if $K \cdot k' = \tfrac{1}{2}\|k'\|^2$, and so by the proof of Proposition \ref{prop:brillouinmult}, $K-k'$ lies on the same number of (hyper)faces of $\cB$ as $K$ does. Together with the fact that $K \in V(\cB)$, this implies $K - k' \in V(\cB)$ as well.
\end{proof}

\end{appendix}

\printbibliography[heading=bibintoc]

\end{document}